\documentclass{article}

\PassOptionsToPackage{numbers, compress}{natbib}

\usepackage[preprint]{neurips2023}

\usepackage[utf8]{inputenc} 
\usepackage[T1]{fontenc}    
\usepackage{hyperref}       
\usepackage{url}            
\usepackage{booktabs}       
\usepackage{amsfonts}       
\usepackage{nicefrac}       
\usepackage{microtype}      
\usepackage{xcolor}         
\usepackage{multirow}
\usepackage{graphicx}
\usepackage{natbib}
\usepackage{caption}
\usepackage{makecell}
\usepackage{wrapfig}
\usepackage{amsmath}
\usepackage{amsthm}
\usepackage{mathrsfs}
\usepackage[resetlabels]{multibib}

\newcites{App}{Appendix}
\bibliographystyleApp{unsrtnat}

\newtheorem{definition}{Definition}
\newtheorem{assumption}{Assumption}
\newtheorem{lemma}{Lemma}

\DeclareMathOperator*{\argmin}{arg\,min}

\title{MolFM: A Multimodal Molecular Foundation Model}

\author{%
  Yizhen Luo$^{1}$,  Kai Yang$^{1}$, Massimo Hong$^{1,2}$, Xing Yi Liu$^{1}$, Zaiqing Nie$^{1,}$\thanks{Corresponding author}\\
  $^{1}$Institute of AI Industry Research (AIR), Tsinghua University\\
  $^{2}$Department of Computer Science and Technology, Tsinghua University\\
  \texttt{\{yz-luo22,hongcd21\}@mails.tsinghua.edu.cn}\\
  \texttt{liuxingyi99@gmail.com}\qquad \texttt{\{yangkai,zaiqing\}@air.tsinghua.edu.cn}     \\
}

\begin{document}

\maketitle

\begin{abstract}
  Molecular knowledge resides within three different modalities of information sources: molecular structures, biomedical documents, and knowledge bases. Effective incorporation of molecular knowledge from these modalities holds paramount significance in facilitating biomedical research. However, existing multimodal molecular foundation models exhibit limitations in capturing intricate connections between molecular structures and texts, and more importantly, none of them attempt to leverage a wealth of molecular expertise derived from knowledge graphs. In this study, we introduce MolFM, a multimodal molecular foundation model designed to facilitate joint representation learning from molecular structures, biomedical texts, and knowledge graphs. We propose cross-modal attention between atoms of molecular structures, neighbors of molecule entities and semantically related texts to facilitate cross-modal comprehension. We provide theoretical analysis that our cross-modal pre-training captures local and global molecular knowledge by minimizing the distance in the feature space between different modalities of the same molecule, as well as molecules sharing similar structures or functions. MolFM achieves state-of-the-art performance on various downstream tasks. On cross-modal retrieval, MolFM outperforms existing models with 12.13\% and 5.04\% absolute gains under the zero-shot and fine-tuning settings, respectively. Furthermore, qualitative analysis showcases MolFM's implicit ability to provide grounding from molecular substructures and knowledge graphs. Code and models are available on \url{https://github.com/BioFM/OpenBioMed}.
\end{abstract}

\section{Introduction}
The understanding of molecular properties and functions is of great significance to broad biomedical applications. Molecular knowledge resides within three multimodal information sources, namely molecular structures, biomedical documents and knowledge bases. Recent advances in Vision-and-Language Pre-training (VLP) \cite{lu2019vilbert,jia2021scaling,radford2021learning,li2022supervision,li2021align,dou2022empirical,VL-BERT} have sparked the emergence of pre-trained multimodal molecular foundation models that jointly learn molecular representations from structures and semantically-related texts. These approaches can be categorized as follows: (1) Generative models, exemplified by KV-PLM \cite{zeng2022deep} and MolT5 \cite{edwards2022translation}, which treat the SMILES string of molecules and texts in a unified model with an auto-encoding framework. (2) Contrastive models, including MoMu \cite{su2022molecular} and MoleculeSTM \cite{liu2022multi}, which conduct contrastive learning with the structural and textual representations of molecules. 


Despite their promising advancements, existing multimodal molecular foundation models suffer from the following key limitations: (1) They fail to fully exploit and fuse the available structural and text information. Generative models primarily rely on 1D SMILES strings to capture structural characteristics and therefore lack the ability to interpret complex topological and spatial properties like macrocycles \cite{zeng2022deep, nguyen2021graphdta}. Contrastive models, on the other hand, tend to overlook the intricate connections between text snippets and substructures of molecules. (2) Existing models predominantly focus on local-level domain knowledge from individual molecules and neglect crucial global-level domain knowledge from knowledge bases. In fact, it has been widely accepted that incorporating global-level knowledge including relationships among molecules, target ligands, diseases and other biomedical entities could greatly facilitate biomedical research \cite{callahan2020knowledge, nicholson2020constructing, luo2023empowering}. 

In this work, we propose MolFM, a multimodal molecular foundation model, to address the aforementioned problems. We aim to conduct joint molecular representation learning that captures both the local knowledge between molecular structures and biomedical texts, as well as the global knowledge from knowledge bases. To accomplish this goal, we first encode 2D molecular graphs, biomedical texts and knowledge graphs independently with pre-trained single-modal encoders. Then, we introduce a multimodal encoder to holistically fuse the features with cross-modal attention between atoms of the molecular structure, neighbors within the knowledge graph and tokens in the textual description. We incorporate structure-text contrastive (STC), cross-modal matching (CMM), masked language model (MLM) and knowledge graph embedding (KGE) as pre-training objectives. More importantly, we provide theoretical justifications that our multimodal pre-training could be interpreted as minimizing the distance in the feature space between different modalities of the same molecule, as well as between molecules that share similar structures or functions. 

We manifest the outstanding performance of MolFM on various downstream tasks. On cross-modal retrieval \cite{edwards2021text2mol, zeng2022deep, su2022molecular}, MolFM achieves absolute gains of 12.13\% and 5.04\% under zero-shot and fine-tuning settings, respectively, compared to the state-of-the-art method MoMu \cite{su2022molecular}. On molecule captioning and text-based molecule generation \cite{edwards2022translation}, we show that MolFM generates more accurate molecules and text descriptions through quantitative and qualitative studies. On molecular property prediction \cite{wu2018moleculenet}, MolFM boosts the prediction performance by 1.55\% absolute gain on average by incorporating multimodal data. We also provide visualization of cross-modal attention, which reveals MolFM's potential to perform grounding based on molecular sub-structures and knowledge graphs.

Our contributions are summarized as follows: (1) We propose MolFM, a multimodal molecular foundation model designed to facilitate joint representation learning from molecular structures, biomedical texts, and knowledge graphs through fine-grained cross attention between different modalities. (2) We theoretically justify that our pre-training approach implicitly minimizes the distance in the feature space between different modalities of the same molecule, as well as between molecules with similar structures or functions. (3) We show the state-of-the-art performance of MolFM on various downstream tasks, thereby highlighting its efficacy and versatility.
\section{Related works}
Our work is connected to the following research topics:

\textbf{Molecular foundation models.} Due to insufficient supervised data in the biomedical domain, molecular foundation models that conduct pre-training on large-scale unsupervised molecules have been developed. Most existing works primarily focus on a single modality of molecules. One line of research aims to learn molecular knowledge from structural representations such as 1D SMILES strings \cite{chithrananda2020chemberta, irwin2022chemformer}, 2D molecular graphs \cite{hustrategies,you2020graph,wang2022molecular,graphmvp} or 3D geometry views \cite{liu2021pre, zhu2022unified, stark20223d}. Another line attempts to implicitly capture molecular expertise through comprehending biomedical literature \cite{beltagy2019scibert, lee2020biobert, gu2021domain, wei2016assessing}. 

More recently, several multimodal approaches \cite{zeng2022deep, edwards2022translation, su2022molecular, liu2022multi} that jointly learn molecular representations from molecular structures and biomedical texts have been proposed. For example, KV-PLM \cite{zeng2022deep} and MolT5 \cite{edwards2022translation} treat SMILES strings and texts as two different languages and perform pre-training with auto-encoding objectives \cite{devlin2019bert, raffel2020exploring}. MoMu \cite{su2022molecular} and MoleculeSTM \cite{liu2022multi} encode molecular graphs and texts with independent encoders and conduct cross-modal contrastive learning \cite{radford2021learning, li2022supervision}. Different from these models, MolFM connects molecular expertise from three modalities, namely structures, texts and knowledge graphs, enabling a more holistic understanding of molecules.  

\textbf{Knowledge-empowered deep learning for molecules.} The incorporation of domain knowledge has shown significant efficacy in various molecule-related tasks, including drug-drug interaction prediction \cite{zhang2017predicting}, drug-target binding affinity prediction \cite{thafar2020dtigems+, ye2021unified}, and molecular property prediction \cite{, luo2023empowering}. However, there are only a few attempts in knowledge-enhanced molecular foundation models. Existing works include MoCL \cite{sun2021mocl} that employs substructure perturbation knowledge and structural similarity knowledge to generate positive samples for contrastive learning. Similarly, KCL \cite{fang2022molecular} augments 2D molecular graphs with the guidance of chemical element knowledge. In contrast, MolFM treats knowledge graphs as an additional input modality instead of a tool to generate structural augmentations. Furthermore, MolFM focuses on capturing richer global knowledge of molecules, such as their relationships with other compounds, target ligands or diseases.

\begin{figure}[tpb]
\centering
    \includegraphics[width=\linewidth]{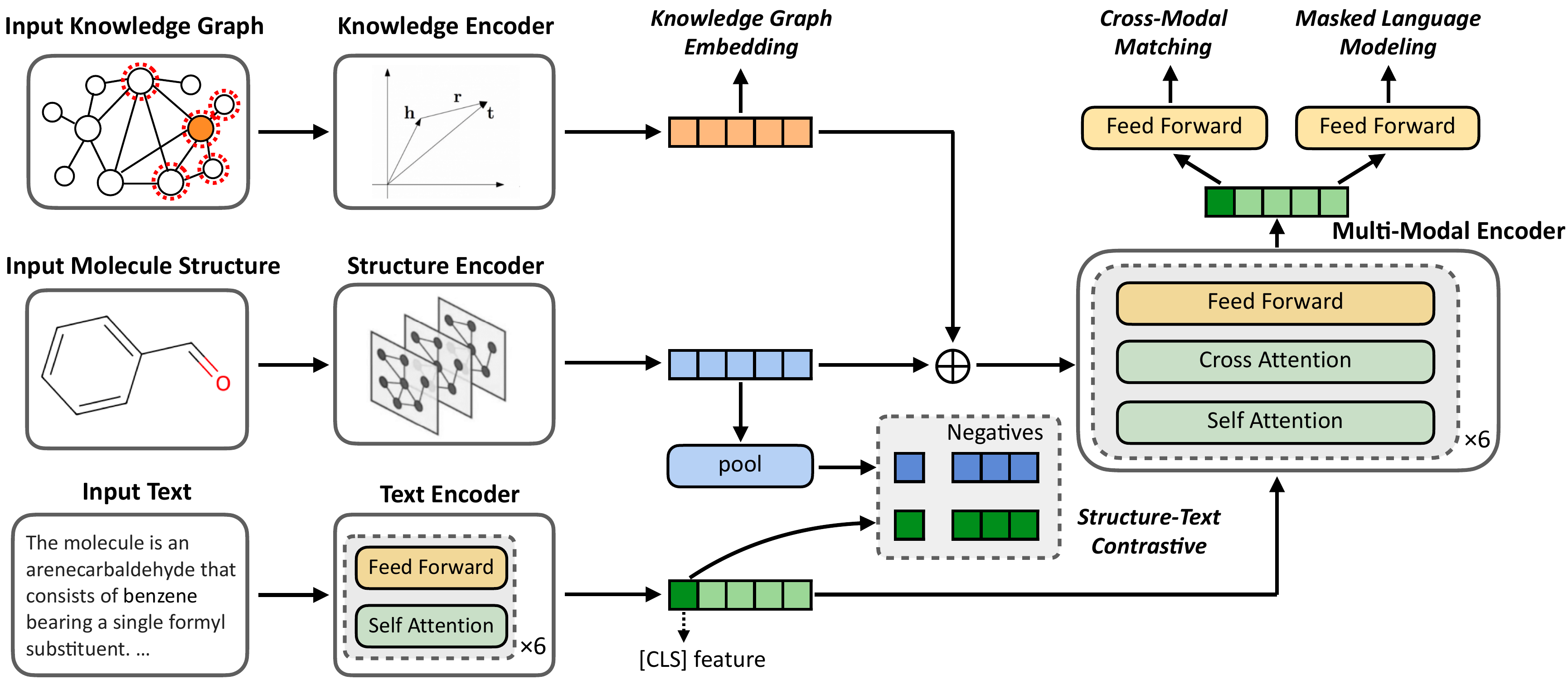}
\captionsetup{font={small,stretch=0.95}}
\caption{\textbf{Pre-training pipeline of MolFM.} We formulate the knowledge graph input for each molecule (dashed circle) as the corresponding entity (orange node) and its 1-hop neighbors. MolFM employs three independent single-modal encoders to convert multimodal inputs into feature vectors. Additionally, it comprises a multimodal encoder to integrate fine-grained connections between atoms, neighboring entities and textual tokens. We leverage structure-text contrastive learning to align the feature space between two modalities, cross-modal matching loss and masked language modeling loss to promote a holistic understanding of multimodal information, and a knowledge embedding loss as a regularization term.}
\label{fig:pipeline}
\end{figure}
\section{MolFM Pre-training}
In this section, we start with a brief introduction to our model architecture (Sec. \ref{sec:model}), followed by the multimodal pre-training objectives (Sec. \ref{sec:pretrain}). Then, we provide theoretical justifications for our approach from the perspective of deep metric learning (Sec. \ref{sec:theory}). Finally, we describe our pre-training dataset and knowledge graph (Sec. \ref{sec:pretrain_data}), as well as implementation details (Sec. \ref{sec:implementation}).

\subsection{Model architecture\label{sec:model}}
The model architecture of MolFM is illustrated in Fig. \ref{fig:pipeline}. MolFM aims to learn a joint representation from molecular structure $S$, biomedical text $T$ and knowledge graph input $K$. We formalize $S$ as a 2D molecular graph $\mathcal{G}=(\mathcal{V}, \mathcal{E})$ where $\mathcal{V}$ represents atoms and $\mathcal{E}$ represents bonds, and $T$ as a sequence of $L$ tokens. We define the overall knowledge graph $KG$ as a graph containing entities as nodes and relations as edges. $KG$ is represented by a set of triplets $\{(h, r, t)\}$ where $h$ and $t$ are head and tail entities, and $r$ is the relation type. Considering that biomedical texts often contain co-occurring mentions of entities related to the molecule \cite{wei2016assessing}, we formulate $K$ as the corresponding molecular entity in $KG$ and $N$ randomly sampled entities from its one-hop neighbors. In this way, $K$ comprises richer information from knowledge graphs to facilitate further multimodal pre-training.

MolFM utilizes three independent encoders pre-trained on single-modality data to encode inputs from different modalities. The molecular graph encoder employs a 5-layer GIN \cite{xu2018powerful} initialized with the weights from GraphMVP \cite{graphmvp} to obtain node representations $h_{SA}$ for atoms and a graph representation $h_{SM}$ for the entire molecule. The text encoder adopts a 6-layer transformer \cite{vaswani2017attention} initialized with the first 6 layers of KV-PLM \cite{zeng2022deep} to generate token features $h_T$. The knowledge graph encoder implements a TransE \cite{bordes2013translating} model, which has been trained on $KG$ for 500 epochs, to compute knowledge features $h_{K}$ for each entity in $K$.

Inspired by \cite{li2021align}, we introduce a multimodal encoder composed of 6 transformer layers with cross attention at each layer. The multimodal encoder is initialized as the last 6 layers of KV-PLM. The cross attention module performs multimodal fusion using token features $h_{T}$ as queries and the concatenation of atom features $h_{SA}$ and neighbor features $h_{K}$ as keys and values.

\subsection{Pre-training objectives\label{sec:pretrain}}
Our pre-training procedure contains 4 objectives: structure-text contrastive loss (STC), cross-modal matching (CMM), masked language modeling (MLM) and knowledge graph embedding (KGE). 

\textbf{Structure-text contrastive loss} aims to align the feature space of structure and text encoders and further facilitate multimodal understanding. We apply fully-connected layers and $l2$ normalization to obtain structural representation $z_S$ from $h_{SM}$ and textual representation $z_T$ from $h_T^{[cls]}$ (the textual feature of the \textit{[CLS]} token). Then, we optimize the following cross-modal contrastive loss \cite{radford2021learning}:
\begin{equation}
    \mathcal{L}_{stc}=-\frac{1}{2}\left[\log\frac{\exp(s(z_S,z_T)/\tau)}{\sum_{S'\in B}\exp(s(z_{S'},z_T)/\tau)}+\log\frac{\exp(s(z_S, z_T)/\tau)}{\sum_{T'\in B}\exp(s(z_S,z_{T'})/\tau)}\right],
\end{equation}
where $s(\cdot, \cdot)$ refers to cosine similarity, $B$ consists of molecular structures and texts within the same mini-batch, and $\tau$ is a temperature hyper-parameter. 

\textbf{Cross-modal matching loss} aims to promote a deeper understanding of molecules by predicting whether the structure, text and knowledge graph data correspond to the same molecule. We randomly permute the multimodal inputs in the mini-batch to create negative samples. We obtain the representation of the $[CLS]$ token from our multimodal encoder $\mathcal{M}_\theta$, and feed it into the predictor $p_{cmm}$ composed of a fully-connected layer and softmax activation. CMM optimizes the following loss:
\begin{equation}
\label{eq:cmm}
    \mathcal{L}_{cmm}=\sum_{(\tilde{S}, \tilde{T}, \tilde{K})\in \tilde{B}}H\left[y_{cmm}(\tilde{S},\tilde{T},\tilde{K}),\ p_{cmm}(\mathcal{M}_{\theta}(h_{\tilde{S}}, h_{\tilde{T}}, h_{\tilde{K}}))\right],
\end{equation}
where $\tilde{B}$ is the corrupted mini-batch with $y_{cmm}(\tilde{S}, \tilde{T}, \tilde{K})$ indicating whether the multimodal data from the mini-batch correspond to the same molecule. $H(\cdot, \cdot)$ denotes cross entropy.

\textbf{Masked language modeling loss} aims to predict the masked tokens using information from three modalities. We adopt the same masking strategy as BERT \cite{devlin2019bert} to generate the masked text $\hat{T}$, and minimize the following objective:
\begin{equation}
    \mathcal{L}_{mlm}=H[y_{mlm}(\hat{T}), p_{mlm}(\mathcal{M}_{\theta}(h_S, h_{\hat{T}}, h_K))],
\end{equation}
where $p_{mlm}$ predicts the probability for masked tokens, and $y_{mlm}(\hat{T})$ is the one-hot ground truth.

\textbf{Knowledge graph embedding loss} serves as a regularization term to prevent the knowledge graph representations from catastrophic forgetting \cite{kirkpatrick2017overcoming}. We randomly sample a positive triplet $(h, r, t)$ from $KG$ for each entity $h$ in $K$. Then we generate two negative triplets $(h, r, \tilde{t})$ and $(\tilde{h}, r, t)$ by randomly sampling $\tilde{t},\tilde{h}$ from all entities and optimize the following max-margin loss: 
\begin{equation}
    \label{eq:kge}
    \mathcal{L}_{kge}=\sum_{h\in K}\left[\max(0, d(h,r,t)-d(h,r,\tilde{t})+\Delta)+\max(0, d(h,r,t)-d(\tilde{h},r,t)+\Delta)\right],
\end{equation}
where $d(h,r,t)=\|f(h)+g(r)-f(t)\|_2$, and $\Delta$ is a margin hyper-parameter. We use $f$ and $g$ to denote embedding functions for entities and relations of our TransE model.

MolFM pre-training optimizes the sum of the aforementioned objectives where $\mathbb{E}[\cdot]$ is expectation:
\begin{equation}
    \mathcal{L}=\mathbb{E}_{(S,T,K)}\left[\mathcal{L}_{stc}+\mathcal{L}_{cmm}+\mathcal{L}_{mlm}+\mathcal{L}_{kge}\right].
\end{equation}
\subsection{Theoretical justifications\label{sec:theory}}
The relationship between conventional multimodal pre-training objectives (STC and MLM) and mutual information maximization has been studied in previous works \cite{oord2018representation, li2021align}. In this section, we interpret CMM and KGE from the perspective of deep metric learning \cite{hoffer2015deep, kaya2019deep} with a brief introduction to our major findings, and defer readers to Appendix \ref{sec:kge} for detailed proofs.

\textbf{CMM learns a fine-grained metric between the multimodal representations of the same molecule.}  We show that $\mathcal{L}_{cmm}$ in Eq. \ref{eq:cmm} satisfies the following:
\begin{equation}
\label{eq:cmm_rewrite}
    \mathcal{L}_{cmm}\propto \sum_{(\tilde{S},\tilde{T},\tilde{K})\in \tilde{B}}\left[-p_{cmm}(\mathcal{M}_\theta(h_S,h_T,h_K)) + p_{cmm}(\mathcal{M}_\theta(h_{\tilde{S}}, h_{\tilde{T}}, h_{\tilde{K}}))\right].
\end{equation}
Eq. \ref{eq:cmm_rewrite} conceptualizes that the multimodal encoder and the CMM predictor compose a scoring function which assigns higher scores to matched structure-text-knowledge triplets and lower scores for unmatched triplets. Therefore, we conclude that CMM further aligns the feature space of three modalities and captures the intrinsic connections between multimodal features.

\textbf{KGE minimizes the distance between molecules sharing similar structures and functions.} We formulate the max-margin loss in Eq. \ref{eq:kge} as a function of the positive triplet $(h,r,t)$. Then, we present two lemmas for structurally and functionally similar molecules in the following:

\begin{lemma}
    \label{lem:structure}
    Let $r_{s}$ be a \textbf{symmetric} relation indicating structural similarity. Assuming that structurally similar molecules $h$ and $t$ satisfies $(h,r_{s},t)\in KG$ and $(t,r_{s},h)\in KG$, the following holds:
    \begin{equation}
        \mathcal{L}_{kge}(h,r_s,t)\propto 2\|f(h)-f(t)\| - \|f(h)-f(\tilde{t})\| - \|f(\tilde{h})-f(t)\|.
    \end{equation}
\end{lemma}

\begin{lemma}
    \label{lem:function}
    Assuming that for functionally similar molecules $h$ and $t$, there exists some entity $o$ and relation $r$ that satisfies $(h,r,o)\in KG, (t, r, o)\in KG$ \textbf{or} $(o,r,h)\in KG, (o,r,t)\in KG$. We use $\mathcal{I}$ to denote the triplets between $h,t$ and these intermediate entities $o$. The following holds:
    \begin{equation}
        \|f(h)-f(t)\|\le \alpha \mathbb{E}_{(e_1,r,e_2)\sim \mathcal{I}}\left[\mathcal{L}_{kge}(e_1,r,e_2)\right] + C,
    \end{equation}
    where $\alpha\approx 1$ and $C\approx 0$ are constants.
\end{lemma}

Lemma \ref{lem:structure} shows that $\mathcal{L}_{kge}$ pulls close the entity embeddings of structurally similar molecules and pushes away dissimilar molecules. In Lemma \ref{lem:function} we hypothesize that functionally similar molecules tend to interact with the same entity (e.g. treats the same disease). Then, we show that the mean $\mathcal{L}_{kge}$ over $\mathcal{I}$ serves as an upper bound for the distance between functionally similar molecules. Hence, by combining CMM and KGE, we empower our multimodal encoder with local knowledge from molecular structures and texts, as well as global knowledge from knowledge graphs. 

\subsection{Pre-training dataset and knowledge graph\label{sec:pretrain_data}}
We follow the pre-training data in \cite{su2022molecular}, which consists of 15K molecules from PubChem \cite{kim2016pubchem} and 37M paragraphs from S2ORC \cite{lo2020s2orc}. We construct our knowledge graph using public databases \cite{wishart2018drugbank, gilson2016bindingdb,delmas2021building} and heuristics \cite{sun2021mocl}. The knowledge graph contains a total of 49K entities and 3.2M relations. We present more details in Appendix \ref{sec:app_data}.

\subsection{Implementation details\label{sec:implementation}}
The MolFM model comprises a molecular structure encoder with 1.8M parameters, a text encoder with 61.8M parameters, a knowledge encoder with 12.6M parameters, and a multi-modal encoder with 61.8M parameters. We pre-train MolFM for 300 epochs with a batch size of 128 on 4 NVIDIA A100 GPUs. We use the AdamW \cite{loshchilov2017decoupled} optimizer with a weight decay of $1e^{-4}$. The learning rate is linearly warmed-up to $1e^{-4}$ in the first 2,000 iterations and then decreases to $1e^{-5}$ following a cosine annealing strategy. We set $N=4, \tau=0.1$ and $\Delta=0.2$. 
\section{Downstream tasks \label{sec:downstream}}
\begin{figure*}[tpb]
\centering
    \includegraphics[width=0.9\linewidth]{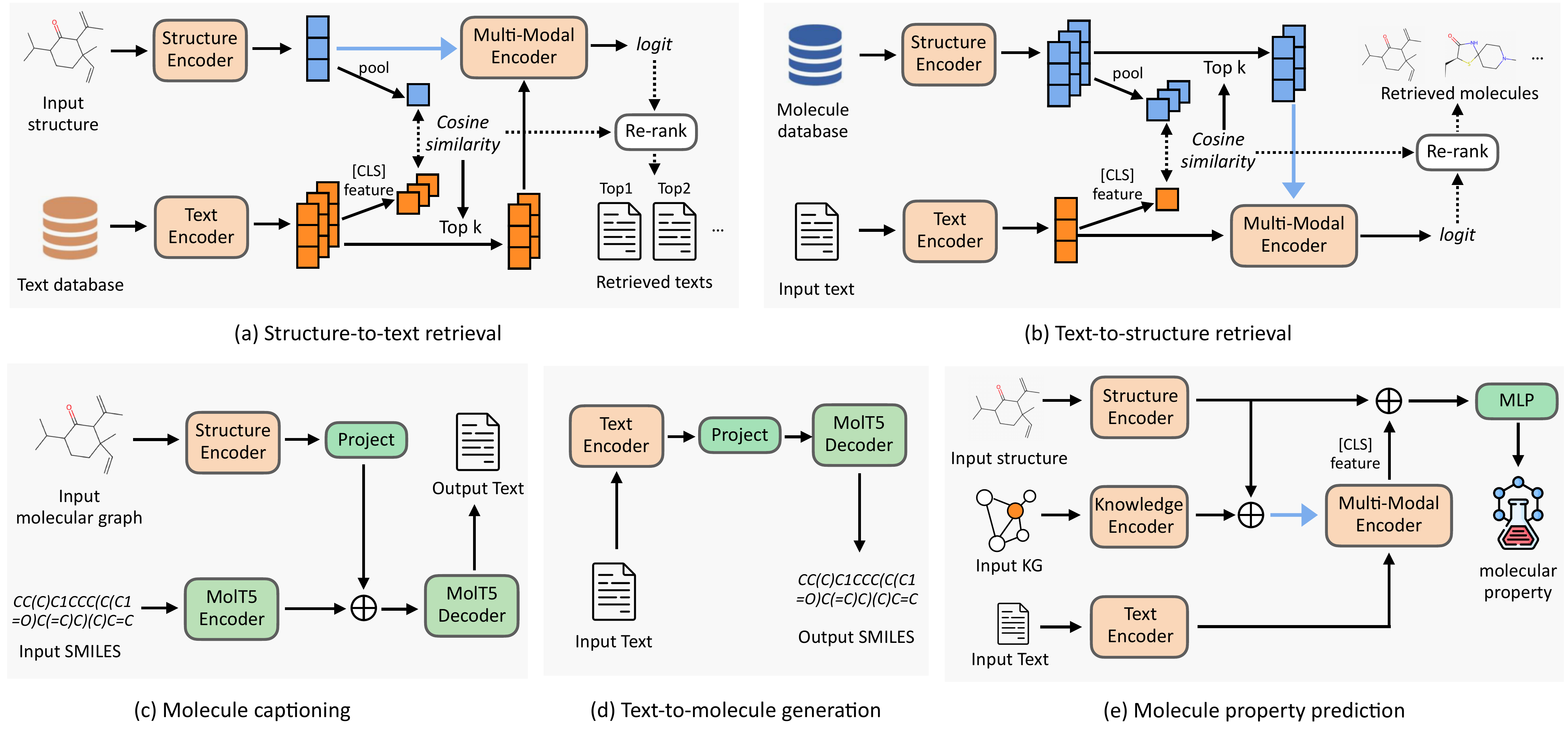}
\captionsetup{font={small,stretch=0.95}}
\caption{\textbf{Model architecture for downstream tasks.} For cross-modal retrieval, we re-rank top-k retrieved results with an ensemble of cosine similarity and CMM logit. For molecule captioning, we concatenate MolFM's structure encoder outputs with MolT5 encoder outputs, and use the MolT5 decoder to generate texts. For text-to-molecule generation, we append a MolT5 decoder to generate SMILES strings. For molecular property prediction, we concatenate the output of structure encoder and multimodal encoder to fit the molecular property.}
\label{fig:downstram_tasks}
\end{figure*}
In this section, we present 4 downstream tasks and their fine-tuning strategy.

\textbf{Cross-modal retrieval} contains two sub-tasks, namely structure-to-text retrieval (S-T) and text-to-structure retrieval (T-S). We evaluate MolFM on PCdes \cite{zeng2022deep} in both zero-shot and fine-tuning scenarios with the entire paragraph as text input. We report MRR (mean reversed rank) and Recall at 1/5/10. As depicted in Fig. \ref{fig:downstram_tasks}a and Fig. \ref{fig:downstram_tasks}b, we modify the re-ranking algorithm in \cite{li2021align} with an ensemble technique \cite{edwards2021text2mol}. Specifically, we simultaneously optimize the fine-tuning objective and CMM loss in Eq. \ref{eq:cmm} during fine-tuning. For inference, we first retrieve the top-$k$ candidates based on cosine similarity. Then, we calculate the CMM logits for these $k$ candidates. Finally, we re-rank them by a linear combination of cosine similarities and CMM logits.

\textbf{Molecule captioning} involves generating descriptions based on molecular structures. We conduct experiments on the ChEBI-20 dataset \cite{edwards2021text2mol} and follow the evaluation metrics in \cite{edwards2022translation}. As shown in Fig. \ref{fig:downstram_tasks}c, we apply a fully-connected layer to project the atom features $h_{SA}$ and concatenate the results with outputs from the MolT5 \cite{edwards2022translation} encoder. Then, we use the MolT5 decoder to generate the caption. 

\textbf{Text-based molecule generation} refers to the task of generating the SMILES strings of molecules using textual descriptions as input. Once again, we utilize the ChEBI-20 dataset and evaluation metrics in \cite{edwards2022translation}. As illustrated in Fig. \ref{fig:downstram_tasks}d, we pass the text features $h_T$ through a fully-connected layer and feed them into the MolT5 decoder to generate SMILES strings.

\textbf{Molecular property prediction} is a vital task in AI-assisted drug discovery. We adopt MoleculeNet \cite{wu2018moleculenet}, a widely recognized benchmark encompassing 8 classification datasets whose prediction objectives range from bio-activity to toxicity. We report ROC\_AUC for each dataset. The prediction pipeline is illustrated in Fig. \ref{fig:downstram_tasks}e. Inspired by DeepEIK \cite{luo2023empowering}, we first obtain knowledge and text data for molecules within the dataset through SMILES matching. Then, we feed the multimodal inputs into MolFM. We concatenate the structure feature $h^{SM}$ with the \textit{[CLS]} feature of the multimodal encoder. Finally, the multimodal feature is passed into a prediction head to fit the molecular property. 

\section{Experiments \label{sec:exp}}
In this section, we first conduct ablation studies to analyze the contributions of different components in MolFM (Sec. \ref{sec:ablation}). Then, we present the state-of-the-art performance of MolFM on cross-modal retrieval (Sec. \ref{sec:cross_modality}), molecule captioning (Sec. \ref{sec:molcap}), text-to-molecule generation (Sec. \ref{sec:t2m}) and molecular property prediction (Sec. \ref{sec:molprop}). Furthermore, we showcase the implicit ability of our model to provide groundings through visualization of cross-modal attention (Sec. \ref{sec:case}). 

\begin{wraptable}{r}{0.47\linewidth}
\small
\centering
\captionsetup{font={small,stretch=0.95}}
\caption{Influence of MolFM components for zero-shot cross-modal retrieval. We report the average of R@1, R@5 and R@10. w/o knowledge: the knowledge graph input is removed. CMM: cross-modal matching. KGE: knowledge graph embedding.}
\label{tab:ablation_pretrain}
\begin{tabular}{lcc}
\toprule
Method                 & S-T   & T-S   \\
\midrule
MolFM                   & \textbf{26.27} & \textbf{28.78} \\
\midrule
- w/o re-rank            & 25.22 & 28.13 \\
- w/o attention to atoms & 23.45 & 25.89 \\
- w/o attention to neighbors & 25.23 & 28.49 \\
- w/o knowledge          & 24.66 & 27.33 \\
\midrule
- w/o KGE                & 25.81 & 28.24 \\
- w/o CMM                & 23.48 & 25.96 \\
- w/o knowledge+CMM      & 22.07 & 24.48 \\
\bottomrule
\end{tabular}
\end{wraptable}
\subsection{Ablation studies\label{sec:ablation}}
To demonstrate the effectiveness of each component in MolFM, we compare performance on zero-shot cross-modal retrieval with different variants of our method in Tab. \ref{tab:ablation_pretrain}. We find that the application of re-ranking improves the retrieval performance. Surprisingly, the performance drops sharply when cross-modal attention to atoms or CMM is removed. These results highlight the significance of learning intricate connections between substructures and word snippets with a multi-modal encoder through appropriate pre-training tasks. Besides, incorporating knowledge graphs yields an average improvement of 1.5\% for the same pre-training tasks, which demonstrates the effectiveness of global molecular knowledge. Furthermore, both attention to neighbors and KGE contributes slightly to MolFM's capability to leverage knowledge graphs.  
\subsection{Evaluation on cross-modal retrieval\label{sec:cross_modality}}
\begin{table}[tpb]\small
\caption{Paragraph-level cross-modal retrieval results on the test split of PCdes.}
\label{tab:retrieval}
\setlength\tabcolsep{4.5pt}
\centering
\begin{tabular}{llcccccccc}
\toprule
\multirow{2}{*}{Mode} & \multirow{2}{*}{Model} & \multicolumn{4}{c}{S-T}                 & \multicolumn{4}{c}{T-S}                 \\
                     &  & MRR   & R@1 & R@5 & R@10 & MRR   & R@1 & R@5 & R@10 \\ \midrule
\multirow{2}{*}{zero-shot} & MoMu \cite{su2022molecular}                   & 9.89 & 5.08    & 12.82  & 18.93     & 10.33 & 4.90    & 14.48  & 20.69     \\
& MolFM              & \textbf{21.42} & \textbf{13.90}    & \textbf{28.69}    & \textbf{36.21}     & \textbf{23.63} & \textbf{16.14}    & \textbf{30.67}    & \textbf{39.54}     \\
\midrule
\multirow{6}{*}{fine-tune} & SciBERT \cite{beltagy2019scibert}  & 24.98 & 16.32 & 33.91 & 42.64 & 23.92 & 14.97 & 34.05 & 41.74 \\
& KV-PLM \cite{zeng2022deep} & 27.41 & 18.35 & 37.15 & 45.43 & 25.97 & 16.55 &  35.85 & 44.75  \\
& KV-PLM* \cite{zeng2022deep} & 29.15 & 20.60 & 37.87 & 45.74 & 28.12 & 19.29 & 37.33 & 45.29 \\
& GraphMVP \cite{graphmvp} & 31.57 & 23.26 & 40.21 & 47.39 & 30.93 & 21.94 & 40.28 & 47.90 \\
& MoMu \cite{su2022molecular}  & 34.29 & 24.47   & 45.38  & 53.84    & 34.53 & 24.87    & 44.93 & 54.25       \\
& MolFM     & \textbf{39.56} & \textbf{29.76}    & \textbf{50.53}    & \textbf{58.63}     & \textbf{39.34} & \textbf{29.39}    & \textbf{50.26}    & \textbf{58.49} \\ 
\bottomrule
\end{tabular}
\end{table}

Tab. \ref{tab:retrieval} shows the overall cross-modal retrieval performance. Detailed results and analysis could be found in Appendix \ref{sec:app_exp} and Appendix \ref{sec:app_case_mtr}. In the zero-shot setting, MolFM achieves a notable increase of 11.08\% and 13.19\% in MRR over the state-of-the-art method MoMu on S-T and T-S retrieval. In the fine-tuning setting, MolFM continues to deliver significant improvements. Given the limited scale and substantial noise of our pre-training dataset, we conclude that MolFM exhibits strong generalization capabilities in cross-modal retrieval tasks.

\subsection{Evaluation on molecule captioning\label{sec:molcap}}
Tab. \ref{tab:molcap} reports the results of molecule captioning, where MolFM consistently achieves state-of-the-art performance. Compared to MolT5 and MoMu, MolFM shows significant advancements in BLEU \cite{papineni2002bleu} and Text2Mol \cite{edwards2021text2mol} measures, indicating that it generates smoother and more semantically related descriptions. In comparison to GraphMVP, MolFM also exhibits modest improvements, demonstrating that our multimodal pre-training further brings benefits to our structure encoder. Additionally, we provide molecule captioning examples in Fig. \ref{fig:molcap} and Appendix \ref{sec:app_case_mp}. It is evident that MolFM shows better understanding of complex functional groups such as oligosaccharides and molecular properties such as inhibitory effects. 
\begin{table}[tpb]\small
\centering
\caption{Molecule captioning results on the test split of ChEBI-20}
\label{tab:molcap}
\setlength\tabcolsep{2pt}
\begin{tabular}{llccccccc}
\toprule
Decoder     & Encoder     & BLEU-2 & BLEU-4 & ROUGE-1 & ROUGE-2 & ROUGE-L & METEOR & Text2Mol \\
\midrule
\multirow{4}{*}{\makecell[l]{MolT5\\-small}} & MolT5-small \cite{edwards2022translation}  & 0.519  & 0.436  & 0.620   & \textbf{0.469}   & 0.563   & 0.551  & 0.540    \\
& MoMu \cite{su2022molecular} & 0.532  & 0.445  & 0.621   & \textbf{0.469}   & \textbf{0.564}   & 0.557  & 0.543    \\
& GraphMVP \cite{graphmvp} & 0.540  & 0.449  & 0.619   & 0.465   & 0.560   & 0.562  & 0.553    \\
& MolFM       & \textbf{0.542}  & \textbf{0.452}  & \textbf{0.623}   & \textbf{0.469}   & 0.562   & \textbf{0.564}  & \textbf{0.557}    \\
\midrule
\multirow{4}{*}{\makecell[l]{MolT5\\-base}} & MolT5-base \cite{edwards2022translation}   & 0.540  & 0.457  & 0.634   & 0.485   & 0.578   & 0.569  & 0.547    \\
& MoMu \cite{su2022molecular}  & 0.549  & 0.462  & 0.630       & 0.479       & 0.575   & 0.576  & 0.558    \\
& GraphMVP \cite{graphmvp}   & 0.577  & 0.491  & 0.651       & 0.505       & 0.592   & 0.599  & 0.570        \\
& MolFM    & \textbf{0.585}  & \textbf{0.498}  & \textbf{0.653}   & \textbf{0.508}   & \textbf{0.594}   & \textbf{0.607}  & \textbf{0.576}   \\
\bottomrule
\end{tabular}
\end{table}
\begin{figure*}[tpb]
\centering
    \includegraphics[width=\linewidth]{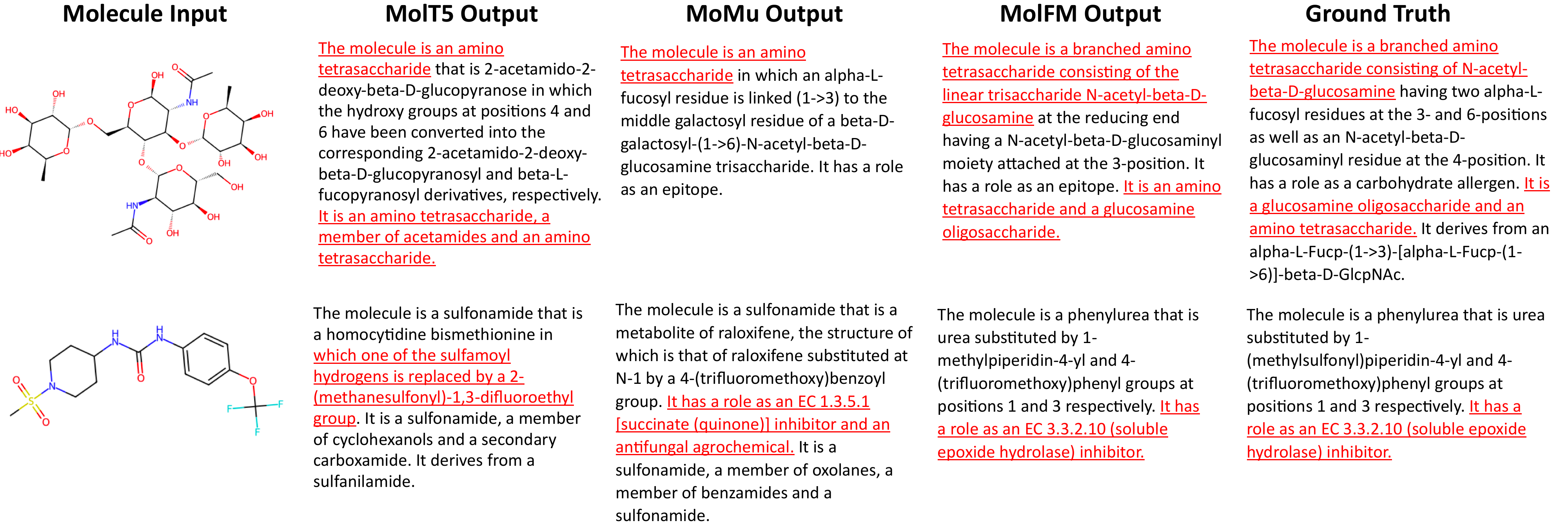}
\captionsetup{font={small,stretch=1.}}
\caption{Molecule captioning examples. We highlight the text segments where MolFM generates more accurate expressions.}
\label{fig:molcap}
\end{figure*}
\subsection{Evaluation on text-to-molecule generation\label{sec:t2m}}
Tab. \ref{tab:t2mgen} shows the results on text-to-molecule generation. MolFM outperforms prior models by generating molecules with considerably higher exact ratio and fingerprint Tanimoto similarity. Qualitative results in Fig. \ref{fig:t2m} also demonstrate that MolFM is able to capture subtle differences between similar sub-structures. Further cases and analysis can be found in Appendix \ref{sec:app_case_t2m}.
\begin{table}[tpb]\footnotesize
\setlength\tabcolsep{3pt}
\centering
\captionsetup{font={small,stretch=0.9}}
\caption{Text-based molecule generation results on the test split of ChEBI-20. $\uparrow$: The higher the better. $\downarrow$: The lower the better.}
\label{tab:t2mgen}
\begin{tabular}{llcccccccc}
\toprule
{\footnotesize Decoder}     & {\footnotesize Encoder}                      & {\footnotesize BLEU $\uparrow$}  & {\footnotesize Exact $\uparrow$} & {\footnotesize Valid $\uparrow$} & {\footnotesize Levenshtein $\downarrow$} & {\footnotesize \makecell[l]{MACCS\\ FTS $\uparrow$}} & {\footnotesize \makecell[l]{RDKit\\ FTS $\uparrow$}} & {\footnotesize \makecell[l]{Morgan\\ FTS $\uparrow$}} & {\footnotesize Text2Mol $\uparrow$} \\
\midrule
\multirow{4}{*}{\makecell[l]{MolT5\\-small}} & MolT5-small \cite{edwards2022translation} & 0.749 & 0.081 & 0.724 & 29.160      & 0.780     & 0.653   & 0.601      & 0.533    \\
& SciBERT \cite{beltagy2019scibert}   & 0.797 & 0.142 & 0.846 & 22.027      & 0.818     & 0.695   & 0.639      & 0.561    \\
& MoMu \cite{su2022molecular}    & 0.800 & 0.150 & 0.858 & 21.446      & 0.818     & 0.709   & 0.651      & 0.566    \\
& MolFM   & \textbf{0.803} & \textbf{0.169} & \textbf{0.859} & \textbf{20.868}      & \textbf{0.834}     & \textbf{0.721}   & \textbf{0.662}      & \textbf{0.573}    \\
\midrule
\multirow{4}{*}{\makecell[l]{MolT5\\-base}} & MolT5-base \cite{edwards2022translation}  & 0.779 & 0.082 & 0.786      & 25.188  & 0.787 & 0.661  & 0.601     & 0.543        \\
& SciBERT \cite{beltagy2019scibert}  & 0.812 & 0.179 & 0.852     & 21.192          & 0.844         & 0.733       & 0.678          & 0.575        \\
& MoMu \cite{su2022molecular}  & 0.815 & 0.183 & 0.863     & 20.520           & 0.847         & 0.737       & 0.678          & 0.580        \\
& MolFM    & \textbf{0.822} & \textbf{0.210} & \textbf{0.892} & \textbf{19.445}      & \textbf{0.854}     & \textbf{0.758}   & \textbf{0.697}      & \textbf{0.583}    \\
\bottomrule
\end{tabular}
\end{table}
\begin{figure*}[tpb]
\centering
    \includegraphics[width=\linewidth]{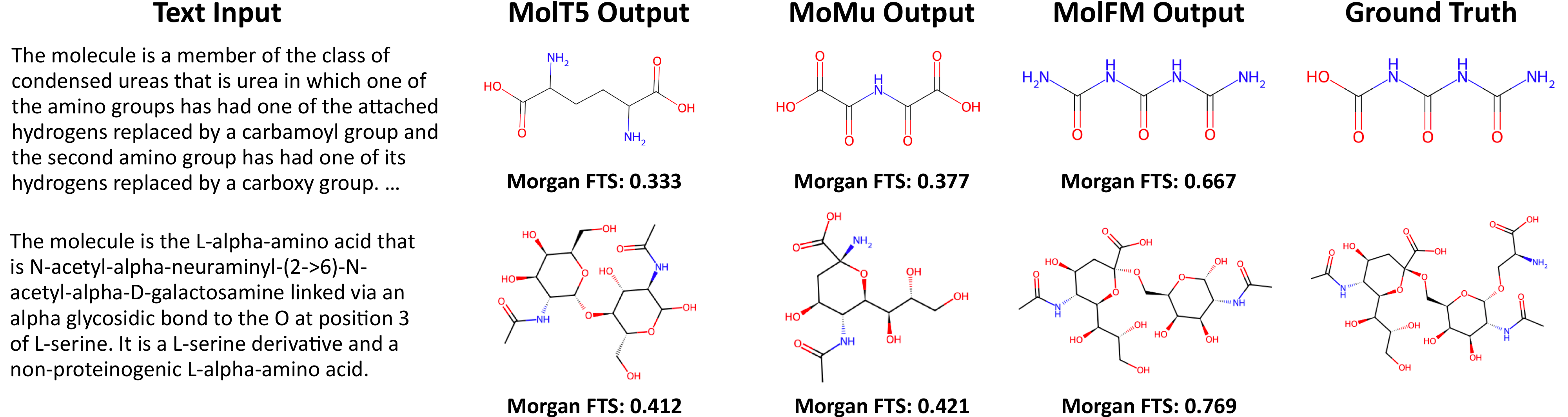}
\captionsetup{font={small,stretch=0.9}}
\caption{Examples of text-to-molecule generation examples, along with the Morgan fingerprint Tanimoto similarity between the generated molecules and the ground truth.}
\label{fig:t2m}
\end{figure*}
\subsection{Evaluation on molecular property prediction\label{sec:molprop}}
Tab. \ref{fig:mp} reports the performance comparison on molecular property prediction. By incorporating additional knowledge graphs and texts, MolFM achieves state-of-the-art performance across 6 out of 8 datasets, demonstrating an average absolute gain of 1.55\% over GraphMVP. When considering inputs from a single modality, namely molecular structure, MolFM shows improved results on Tox21, ToxCast, MUV, HIV and BACE, no statistically significant difference on BBBP and ClinTox, and a slight performance decrease on SIDER compared to GraphMVP. These results highlight the effectiveness of our pre-training, especially when leveraging multimodal information.

\begin{table}[htpb]\small
\setlength\tabcolsep{2.5pt}
\captionsetup{font={small,stretch=1.}}
\caption{Molecular property prediction results on MoleculeNet. w/o T+K: without the additional inputs from texts and knowledge graphs. w/ T+K: with the additional inputs from texts and knowledge graphs.}
\label{fig:mp}
\begin{tabular}{lccccccccc}
\toprule
Model  & BBBP & Tox21 & ToxCast & SIDER & ClinTox & MUV  & HIV  & BACE & Avg  \\
\midrule
GIN             & 65.4$_{\pm 2.4}$ & 74.9$_{\pm 0.8}$  & 61.6$_{\pm 1.2}$    & 58.0$_{\pm 2.4}$  & 58.8$_{\pm 5.5}$ & 71.0$_{\pm 2.5}$ & 75.3$_{\pm 0.5}$ & 72.6$_{\pm 4.9}$ & 67.21 \\
\midrule
AttrMask \cite{hustrategies}      & 70.2$_{\pm 0.5}$ & 74.2$_{\pm 0.8}$  & 62.5$_{\pm 0.4}$    & 60.4$_{\pm 0.6}$  & 68.6$_{\pm 9.6}$    & 73.9$_{\pm 1.3}$ & 74.3$_{\pm 0.6}$ & 77.2$_{\pm 1.4}$ & 70.16 \\
ContextPred \cite{hustrategies}  & 71.2$_{\pm 0.9}$ & 73.3$_{\pm 0.5}$  & 62.8$_{\pm 0.3}$    & 59.3$_{\pm 1.4}$  & 73.7$_{\pm 4.0}$    & 72.5$_{\pm 2.2}$ & 75.8$_{\pm 1.1}$ & 78.6$_{\pm 1.4}$ & 70.89 \\
GraphCL \cite{you2020graph}      & 67.5$_{\pm 3.3}$ & 75.0$_{\pm 0.3}$  & 62.8$_{\pm 0.2}$    & 60.1$_{\pm 1.3}$  & 78.9$_{\pm 4.2}$    & \textbf{77.1$_{\pm 1.0}$} & 75.0$_{\pm 0.4}$ & 68.7$_{\pm 7.8}$ & 70.64 \\
GraphMVP \cite{graphmvp}     & 72.4$_{\pm 1.6}$ & 74.4$_{\pm 0.2}$  & 63.1$_{\pm 0.4}$    & 63.9$_{\pm 1.2}$  & 77.5$_{\pm 4.2}$    & 75.0$_{\pm 1.0}$ & 77.0$_{\pm 1.2}$ & 81.2$_{\pm 0.9}$ & 73.07 \\
\midrule
KV-PLM \cite{zeng2022deep}       & 66.9$_{\pm 1.1}$ & 64.7$_{\pm 1.8}$ & 58.6$_{\pm 0.4}$ & 55.3$_{\pm 0.8}$ & 84.3$_{\pm 1.5}$ & 60.2$_{\pm 2.9}$ & 68.8$_{\pm 4.9}$ & 71.9$_{\pm 2.1}$ & 66.29     \\
DeepEIK \cite{luo2023empowering} & 72.1$_{\pm 0.4}$ & 72.4$_{\pm0.9}$ & 61.5$_{\pm 0.4}$ & 63.5$_{\pm 0.9}$ & \textbf{89.7$_{\pm 1.8}$} & 71.4$_{\pm 1.0}$ & 75.0$_{\pm 0.6}$ & 80.5$_{\pm 1.2}$ & 73.27 \\
MoMu \cite{su2022molecular}         & 70.5$_{\pm 2.0}$ & 75.6$_{\pm 0.3}$  & 63.4$_{\pm 0.5}$    & 60.5$_{\pm 0.9}$  & 79.9$_{\pm 4.1}$  & 70.5$_{\pm 1.4}$ & 75.9$_{\pm 0.8}$ & 76.7$_{\pm 2.1}$ & 71.63 \\
MolFM (w/o T+K) & 72.2$_{\pm 0.1}$ & 76.6$_{\pm 0.4}$ & 64.2$_{\pm 0.1}$ & 63.2$_{\pm 0.3}$ & 78.6$_{\pm 1.3}$ & 76.0$_{\pm 0.8}$ & 78.2$_{\pm 0.4}$ & 82.6$_{\pm 0.6}$  & 73.95     \\
MolFM (w/ T+K)  & \textbf{72.9$_{\pm 0.1}$} & \textbf{77.2$_{\pm 0.7}$}  & \textbf{64.4$_{\pm 0.2}$}       & \textbf{64.2$_{\pm 0.9}$}   & 79.7$_{\pm 1.6}$  & 76.0$_{\pm 0.8}$       & \textbf{78.8$_{\pm 1.1}$}   & \textbf{83.9$_{\pm 1.1}$}  & \textbf{74.62}          \\
\bottomrule
\end{tabular}
\end{table}
\subsection{Visualization of cross-modal attention\label{sec:case}}
\begin{figure*} [tpb]
\centering
    \includegraphics[width=\linewidth]{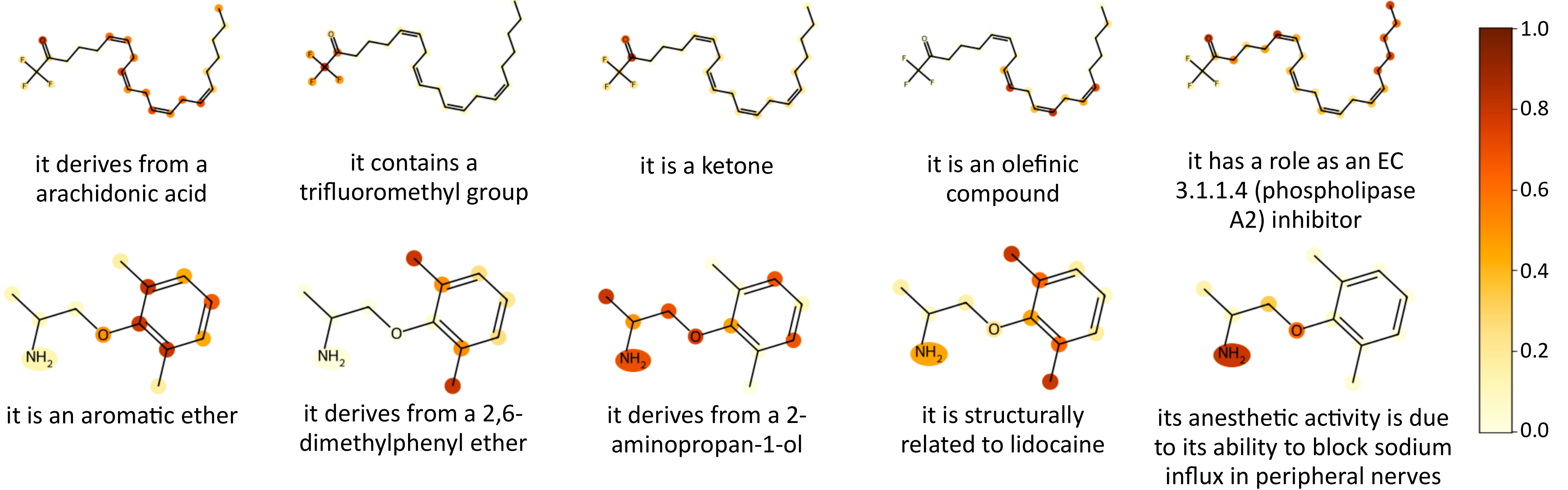}
\captionsetup{font={small,stretch=1.}}
\caption{Visualization of atom attention with different input texts.}
\label{fig:attn_atom}
\end{figure*}
\begin{figure*} [tpb]
\centering
    \includegraphics[width=\linewidth]{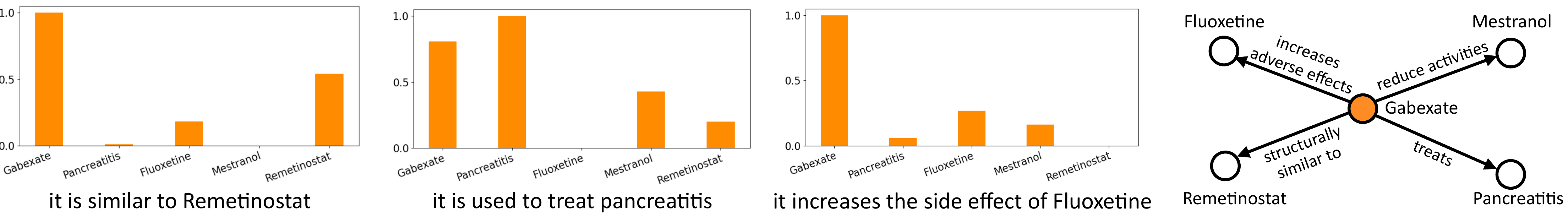}
\captionsetup{font={small,stretch=1.}}
\caption{Visualization of neighbor attention. \textbf{Left}: the input text and the normalized attention value to different entities. \textbf{Right}: the selected molecule (orange) and the relationships with its one-hop neighbors.}
\label{fig:attn_neigh}
\end{figure*}
We provide visualizations of our cross-modal attention between atoms, neighbors and texts in Fig. \ref{fig:attn_atom} , Fig. \ref{fig:attn_neigh} and Appendix \ref{sec:app_vis_attn}. We randomly select molecules and input phrases describing their sub-structures or properties, and display the attention maps of \textit{[CLS]} in the last cross attention layer of the multimodal encoder with a min-max normalization. Notably, the highlighted atoms in Fig. \ref{fig:attn_atom} form substructures that are strongly correlated to the text semantics. The multimodal attention in Fig. \ref{fig:attn_neigh} also captures relevant entities based on textual descriptions. These results reveal the potential of MolFM to establish meaningful associations between structures, texts and knowledge graphs. 
\section{Limitations and broader impacts}
While our work presents promising results in multi-modal molecular modeling, there are still areas for improvement and future exploration: (1) Our pre-training dataset may introduce biases or harmful information to MolFM due to its scale and quality. (2) MolFM may bring limited benefits to newly emerged molecules that lack available text and knowledge information. (3) While MolFM primarily focuses on molecules, incorporating other entities such as proteins, genes, and cell lines may lead to an even more comprehensive understanding of the biomedical context.

MolFM presents significant benefits for accelerating pharmaceutical research by connecting molecular structure with natural language and expert knowledge. However, there is a concern that MolFM may be misused to generate potentially dangerous or toxic molecules. Therefore, it is essential to ensure the responsible and ethical use of the model. We emphasize that MolFM should be employed solely for research purposes, and any further medical applications of MolFM should proceed with caution and undergo comprehensive experimental evaluations.
\section{Conclusion}
In this paper, we present MolFM, a multimodal molecular foundation model to facilitate joint representation learning with molecular structures, biomedical texts and knowledge graphs through leveraging fine-grained cross attention between three modalities. We demonstrate the effectiveness of our pre-training paradigm by both theoretical analysis and experimental evaluation. MolFM achieves state-of-the-art performance on various downstream tasks, with exceptional improvements in cross-modal retrieval. Under thorough analysis aimed at safety, MolFM has the potential to deliver unprecedented benefits to the biomedical research community.

\begin{ack}
This work is supported by the National Key R\&D Program of China (No. 2022YFF1203002).
\end{ack}

\bibliography{ref}

\newpage
\setcounter{section}{0}
\setcounter{equation}{0}
\setcounter{subsection}{0}
\setcounter{table}{0}
\setcounter{figure}{0}
\setcounter{lemma}{0}
\renewcommand{\theequation}{A.\arabic{equation}}
\renewcommand{\thefigure}{A.\arabic{figure}}
\renewcommand{\thetable}{A.\arabic{table}}
\renewcommand{\thelemma}{A.\arabic{lemma}}
\renewcommand{\thedefinition}{A.\arabic{definition}}
\renewcommand{\thesubsection}{\Alph{subsection}}
\section*{Appendix}

\subsection{Theoretical justifications for MolFM pre-training}
In this section, we establish a connection between our pre-training objectives and deep metric learning. We first show that MolFM aligns the feature space for different modalities of the same molecule by analyzing structure-text contrastive (Sec. \ref{sec:stc}), masked language model (Sec. \ref{sec:mlm}) and cross-modal matching (Sec. \ref{sec:cmm}). Then, we give detailed proofs for two lemmas presented in the main document, demonstrating that MolFM grasps global molecular expertise including structural and functional similarity (Sec. \ref{sec:kge}). 

\subsubsection{Analysis of structure-text contrastive (STC) loss \label{sec:stc}}
Given a set of triplets $(x, y, z)$ where $x$ is the anchor sample, $y$ is the positive sample that shares semantic correlations with $x$, and $z$ is the negative sample, deep metric learning \citeApp{_hoffer2015deep, _kaya2019deep} aims to learn a representation network $\mathcal{F}_{\Theta}(\cdot)$ and a distance metric function $\mathcal{D}_{\beta}(\cdot,\cdot)$ that minimizes the distance between $\mathcal{F}_{\Theta}(x)$ and $\mathcal{F}_{\Theta}(y)$ and maximizes the distance between $\mathcal{F}_{\Theta}(x)$ and $\mathcal{F}_{\Theta}(z)$:
\begin{equation}
    \argmin_{\Theta,\beta} \mathbb{E}_{(x,y,z)}\left[\mathcal{D}_{\beta}(\mathcal{F}_\Theta(x), \mathcal{F}_\Theta(y)) - \mathcal{D}_{\beta}(\mathcal{F}_\Theta(x), \mathcal{F}_\Theta(z))\right],
\end{equation}
where $\Theta$ and $\beta$ are model parameters. 

It has been well studied that optimizing the InfoNCE loss is equivalent to maximizing a lower bound of the mutual information between two different views of a data point \citeApp{_oord2018representation, _li2021align}:
\begin{equation}
\label{eq:infonce}
    I(A;B)\ge -\mathcal{L}_{NCE}=\mathbb{E}_{(a,b)}\left[\log\frac{\exp(s(a,b))}{\sum_{\tilde{b}\in \tilde{B}}\exp(s(a,\tilde{b}))}\right],
\end{equation}
where $A, B$ are random variables for the embeddings of different views, $a, b$ are positive samples, $I(\cdot;\cdot)$ denotes mutual information, $s(\cdot,\cdot)$ is a scoring function (we use cosine similarity in this study), and $\tilde{B}$ is a proposal distribution that contains $b$ and $|\tilde{B}|-1$ data points. Following \citeApp{_sun2021mocl}, we connect InfoNCE loss in Eq. \ref{eq:infonce} with deep metric learning by approximating $\log(1+x)$ as $x$ and first-order Taylor expansion:

\begin{equation}
\label{eq:infonce_dml}
    \begin{aligned}
        -\mathbb{E}_{(a,b)}\left[\log \frac{\exp(s(a,b))}{\sum_{\tilde{b}\in \tilde{B}}\exp(s(a,\tilde{b}))}\right] &= \mathbb{E}_{(a,b)}\left[\log\left(1+\sum_{\tilde{b}\in \tilde{B}, \tilde{b}\neq b}\exp(s(a,\tilde{b})-s(a,b))\right) \right]\\
        &\approx \mathbb{E}_{(a,b)}\left[\sum_{\tilde{b}\in \tilde{B}, \tilde{b}\neq b} \exp(s(a,\tilde{b})-s(a, b))\right]\\
        &\propto -\mathbb{E}_{(a,b)}\left[\sum_{\tilde{b}\in \tilde{B}} \left[s(a, b)-s(a, \tilde{b})\right]\right].
    \end{aligned}
\end{equation}
By conceptualizing $-s(\cdot, \cdot)$ as the metric function, the equation above establishes the connection between contrastive learning and deep metric learning.

Hence, assuming the projection head is an identical mapping, our structure-text contrastive (STC) aligns structural and textual representations for the same molecule in the following:

\begin{equation}
\begin{aligned}
    \mathcal{L}_{stc}&=-\frac{1}{2}\mathbb{E}_{(S,T,K)}\left[\log\frac{\exp(s(z_S,z_T)/\tau)}{\sum_{S'\in B}\exp(s(z_{S'},z_T)/\tau)}+\log\frac{\exp(s(z_S, z_T)/\tau)}{\sum_{T'\in B}\exp(s(z_S,z_{T'})/\tau))}\right]\\&\propto -\frac{1}{2\tau}\mathbb{E}_{(S,T,K)}\left[\sum_{S'\in B}\left[s(h_S,h_T) - s(h_{S'}, h_T)\right] + \sum_{T'\in B}\left[s(h_S,h_T) - s(h_S, h_{T'})\right]\right],
\end{aligned}
\end{equation}

where $B$ consists of molecular structures and texts within the same mini-batch, and $\tau$ is a temperature hyper-parameter.

\subsubsection{Analysis of masked language model (MLM)\label{sec:mlm}}

Following \citeApp{_li2021align}, we rewrite masked language modeling loss based on Eq. \ref{eq:infonce_dml} in the following:

\begin{equation}
\begin{aligned}
    \mathcal{L}_{mlm}&=\mathbb{E}_{(S,\hat{T},K)} \left[H(y_{mlm}(\hat{T}), p_{mlm}(\mathcal{M}_{\theta}(h_S, h_{\hat{T}}, h_K))) \right] \\
    &=-\mathbb{E}_{(S,\hat{T},K)} \left[ \log\frac{\exp[s(y_{mlm}(\hat{T}), p_{mlm}(\mathcal{M}_{\theta}(h_S,h_{\hat{T}}, h_K)))]}{\sum_{y\in \mathcal{V}}\exp[s(\psi(y),p_{mlm}(\mathcal{M}_{\theta}(h_S,h_{\hat{T}}, h_K)))]} \right] \\
    &\propto -\mathbb{E}_{(S,\hat{T},K)} \left[ \sum_{y\in \mathcal{V}}[s(y_{mlm}(\hat{T}), p_{mlm}(\mathcal{M}_{\theta}(h_S,h_{\hat{T}}, h_K))) \right.\\& \qquad \qquad \qquad \quad \left. -s(\psi(y), p_{mlm}(\mathcal{M}_{\theta}(h_S,h_{\hat{T}}, h_K)))] \vphantom{\sum_{y\in \mathcal{V}}} \right],
\end{aligned}
\end{equation}
where $\hat{T}$ is the masked token sequence, $y_{mlm}(\hat{T})$ is the one-hot ground truth of the masked token, $\mathcal{M}_{\theta}$ is the multi-modal encoder, $p_{mlm}$ is a predictor that calculates the probability distribution for masked tokens, and $\psi(\cdot):\mathcal{V}\rightarrow \mathbb{R}^{|\mathcal{V}|}$ is a function that maps tokens in the vocabulary set $\mathcal{V}$ to one-hot encodings. Hence, MLM pulls close representations between masked tokens with their multi-modal context.

\subsubsection{Analysis of cross-modal matching (CMM) \label{sec:cmm}}

Based on Eq. \ref{eq:infonce_dml}, the cross-modal matching loss is equivalent to the following:
\begin{equation}
\begin{aligned}
\mathcal{L}_{cmm}&=\mathbb{E}_{(S,T,K)} \left[ \sum_{(\tilde{S}, \tilde{T}, \tilde{K})\in \tilde{B}}H(y_{cmm}(\tilde{S},\tilde{T},\tilde{K}), p_{cmm}(\mathcal{M}_{\theta}(h_{\tilde{S}}, h_{\tilde{T}}, h_{\tilde{K}}))) \right]\\&=-\mathbb{E}_{(S,T,K)} \left[ \log \frac{\exp[p_{cmm}(\mathcal{M}_{\theta}(h_{S}, h_{T}, h_{K}))]}{\sum_{(\tilde{S}, \tilde{T}, \tilde{K})\in \tilde{B}}\exp[(p_{cmm}(\mathcal{M}_{\theta}(h_{\tilde{S}}, h_{\tilde{T}}, h_{\tilde{K}}))]} \right] \\&\propto -\sum_{(\tilde{S},\tilde{T},\tilde{K})\in \tilde{B}}\left[p_{cmm}(\mathcal{M}_\theta(h_S,h_T,h_K)) - p_{cmm}(\mathcal{M}_\theta(h_{\tilde{S}}, h_{\tilde{T}}, h_{\tilde{K}}))\right],
\end{aligned}
\end{equation}
where $\tilde{B}$ is the corrupted mini-batch with $y_{cmm}(\tilde{S}, \tilde{T}, \tilde{K})$ indicating whether the multimodal data from the mini-batch correspond to the same molecule, and $p_{cmm}$ is a binary predictor. By conceptualizing $-p_{cmm}(\mathcal{M}_{\theta}(\cdot,\cdot,\cdot))$ as a distance function, we demonstrate that CMM aligns the structural, texutal and knowledge graph representations of the same molecule.

\subsubsection{Analysis of knowledge graph embedding (KGE) \label{sec:kge}}
In this sub-section, we start with several definitions and denotations with respect to the knowledge graph embedding algorithm. Then we prove the two lemmas in the main document, showing that KGE pulls close embeddings for molecules that shares similar structures (Lemma. \ref{lem:app_structure}) or similar functions (Lemma. \ref{lem:app_function}).

\begin{definition}
Knowledge Graph Embedding. We define $KG=\{(h,r,t)|h,t\in \mathscr{E}, r\in \mathscr{R}\}$ where $\mathscr{E}$ is the entity set and $\mathscr{R}$ is the relation set. We define $N=|\mathscr{E}|$ (the number of entities) and $M=|KG|$ (the number of relations), and use $x\sim X$ to denote that $x$ is uniformly sampled from $X$. KGE aims to learn an entity embedding function $f:\mathscr{E}\rightarrow \mathbb{R}^n$ and a relation embedding function $g:\mathscr{R}\rightarrow \mathbb{R}^n$ by optimizing the following max-of-margin loss for each triplet $(h,r,t)\in KG$:
\begin{equation}
\label{equ:kge}
\begin{aligned}
    \mathcal{L}_{kge}(h,r,t)=&\mathbb{E}_{\tilde{t}\sim \mathscr{E}\backslash t} \left[ \max(0, d(h,r,t)-d(h, r, \tilde{t})+\Delta) \right] \\+&\mathbb{E}_{\tilde{h}\sim \mathscr{E}\backslash h} \left[ \max(0, d(h,r,t)-d(\tilde{h}, r, t)+\Delta) \right],
\end{aligned}
\end{equation}
where $d(h,r,t)=\|f(h)+g(r)-f(t)\|_2$ is a distance function and $\Delta$ is a margin hyper-parameter.
\end{definition}
\begin{definition}
    Given a subset $\mathcal{T}\subset KG$, assume that $\mathcal{X}_{h,r,t}=1$ indicates $(h,r,t)\in \mathcal{T}$ and $\mathcal{X}_{h,r,t}=0$ indicates $(h,r,t)\notin \mathcal{T}$. We define $d_{h,r}^{out}=\sum_{t\in \mathscr{E}}\mathcal{X}_{h,r,t}$ as the out-degree of $h$ under relation $r$ with respect to $\mathcal{T}$, and $d_{t,r}^{in}=\sum_{h\in \mathscr{E}}\mathcal{X}_{h,r,t}$ as the in-degree of $t$ under relation $r$ with respect to $\mathcal{T}$.
\end{definition}
We further give the following assumptions:

\begin{assumption}
\label{assump:delta}
    $\Delta>d(h_1,r_1,t_1)-d(h_2,r_2,t_2)$ for all $h_1,t_1,h_2,t_2\in \mathscr{E}$ and $r_1,r_2\in \mathscr{R}$.
\end{assumption}

\begin{definition}
    Structurally similar molecules. Assume that $h, t\in \mathscr{E}$ are two molecular entities, and $r_s\in \mathscr{R}$ is a relation type indicating structural similarity. $h$ and $t$ are structurally similar if and only if $(h,r_s,t)\in KG$ and $(t,r_s,h)\in KG$. We use $S=\{(h, r_s, t)|(h, r_s, t)\in KG\}$ to denote the set of structural similar relations and assume that $|S|\ge 4$. 
\end{definition}

\begin{assumption}
\label{assump:sym}
Symmetry of $r_s$: $\forall h,t\in \mathscr{E}, (h,r_s,t)\in KG\Leftrightarrow (t,r_s,h)\in KG$.
\end{assumption}

\begin{assumption}
\label{assump:isotropy}
    Isotropy of $r_s$: $f(h)+g(r_s)-f(t)$ is uniformly distributed in all directions for $(h,r_s,t)\in S$ or arbitrary $h,t\in \mathscr{E}$. Further, we use $\alpha_{(h,r_s,t)}$ to denote the angle between an arbitrary vector $x$ and $f(t)+g(r_s)-f(h)$, and hypothesize that for all $x\in \mathbb{R}^d$, the following holds:
    \begin{equation}
    \begin{aligned}
        &-\epsilon \le \sum_{(h,r_s,t)\in S}\frac{\cos 2\alpha_{(h,r_s,t)}}{d(h,r_s,t)}\le \epsilon,\\
        &-\epsilon \le \sum_{h,t\in \mathscr{E}}\frac{\cos 2\alpha_{(h,r_s,t)}}{d(h,r_s,t)}\le \epsilon,
    \end{aligned}
    \end{equation}
    where $\epsilon>0$ is a constant that is close to 0.
\end{assumption}

\begin{assumption}
\label{assump:sparsity}
    Sparsity of $r_s$: $d_{h,r_s}^{out}\le \frac{N}{2}$ and $d_{t,r_s}^{in}\le \frac{N}{2}$ for all $h, t\in \mathscr{E}$.
\end{assumption}

\begin{assumption}
\label{assump:overweigh}
    Distance margin between positive and negative samples. If $(h,r_s,t)\in KG$, $(h,r_s,\tilde{t})\notin KG$ and $(\tilde{h},r_s,t)\notin KG$, the following holds:
    $$\frac{1}{d(h,r_s,t)}-\frac{1}{d(h,r_s,\tilde{t})}\ge \epsilon, \frac{1}{d(h,r_s,t)}-\frac{1}{d(\tilde{h},r_s,t)}\ge \epsilon.$$
\end{assumption}

\begin{lemma}
\label{lem:app_structure}
For structurally similar molecules $h$ and $t$, the following holds:
\begin{equation}
    \mathcal{L}_{kge}(h,r_s,t)\propto 2\|f(h)-f(t)\|_2-\mathbb{E}_{\tilde{t}\sim \mathscr{E}\backslash t}  \|f(h)-f(\tilde{t})\|_2 -\mathbb{E}_{\tilde{h}\sim \mathscr{E}\backslash h}  \|f(\tilde{h})-f(t)\|_2 .
\end{equation}
\end{lemma}

\begin{proof}
    Our proof sketch is showing that optimizing KGE substantially leads to $g(r_s)=0$. Formally: 
    \begin{equation}
        \argmin_{g(r_s)}\mathbb{E}_{(h,r,t)\sim KG} \left[ \mathcal{L}_{kge}(h,r,t) \right] =0.
    \end{equation}
     We first rewrite Eq. \ref{equ:kge} in the following based on Assumption \ref{assump:delta}:
    \begin{equation}
    \begin{aligned}
        \mathcal{L}&=\mathbb{E}_{(h,r,t)\sim KG}\left[\mathcal{L}_{kge}(h,r,t)\right]\\&=\frac{|S|}{M}\sum_{(h,r_s,t)\in S}\left[2d(h,r_s,t)-\mathbb{E}_{\tilde{t}\sim \mathscr{E}\backslash t}\left[d(h,r_s,\tilde{t}) \right] -\mathbb{E}_{\tilde{h}\sim \mathscr{E}\backslash h} \left[ d(\tilde{h},r_s,t) \right] \right]\\
        &\quad +\frac{M-|S|}{M}\sum_{(h,r,t)\in KG\backslash S}\left[2d(h,r,t)-\mathbb{E}_{\tilde{t}\sim \mathscr{E}\backslash t} \left[ d(h,r,\tilde{t}) \right] -\mathbb{E}_{\tilde{h}\sim \mathscr{E}\backslash h} \left[ d(\tilde{h},r,t)\right] \right].
    \end{aligned}
    \end{equation}
    Following \citeApp{_qiu2018revisiting}, we rewrite negative sampling terms as follows:
    \begin{equation}
    \label{equ:neg_tail}
    \begin{aligned}
        \sum_{(h,r,t)\in S}\mathbb{E}_{\tilde{t}\sim \mathscr{E}\backslash t} \left[ d(h,r,\tilde{t}) \right] &=\frac{1}{N-1}\sum_{(h,r,t)\in S}\left[-d(h,r,t)+\sum_{\tilde{t}\in \mathscr{E}}d(h,r,\tilde{t})\right]\\
        &=\sum_{h,t\in \mathscr{E}, r\in \mathscr{R}}\frac{d^{out}_{h,r}}{N-1}d(h,r,t)-\frac{1}{N-1}\sum_{(h,r,t)\in S}d(h,r,t),
    \end{aligned}
    \end{equation}
    and:
    \begin{equation}
    \label{equ:neg_head}
    \begin{aligned}
        \sum_{(h,r,t)\in S}\mathbb{E}_{\tilde{h}\sim \mathscr{E}\backslash h} \left[ d(\tilde{h},r,t) \right] =\sum_{h,t\in \mathscr{E}, r\in \mathscr{R}}\frac{d^{in}_{t,r}}{N-1}d(h,r,t)-\frac{1}{N-1}\sum_{(h,r,t)\in KG}d(h,r,t),
    \end{aligned}
    \end{equation}
    Due to the symmetry of $r_s$ (Assumption \ref{assump:sym}), we can derive that $d_{h,r_s}^{out}=d_{h,r_s}^{in}$ for all $h\in \mathscr{E}$.
    
    As suggested in \citeApp{_bordes2014semantic}, we speculate a value independence for each $r\in \mathscr{R}$ given sufficient large embedding dimension $n$. Hence, we calculate the partial derivative with $g(r_s)$ as follows:
    \begin{equation}
    \label{equ:deriv}
        \begin{aligned}
            \frac{\partial \mathcal{L}}{\partial g(r_s)}
            &=\frac{2|S|N}{M(N - 1)}\sum_{(h,r_s,t)\in S}\frac{\partial d(h,r_s,t)}{\partial g(r_s)}-\frac{|S|} {|M|(N - 1)}\sum_{h,t\in \mathscr{E}}(d_{h,r_s}^{out} + d_{t,r_s}^{in})\frac{\partial d(h,r_s,t)}{\partial g(r_s)}\\
            &=\frac{2|S|N}{M(N - 1)}\sum_{(h,r_s,t)\in S}\frac{f(h)+g(r_s)-h(t)}{d(h,r_s,t)}\\
            &\quad-\frac{|S|} {|M|(N - 1)}\sum_{h,t\in \mathscr{E}}(d_{h,r_s}^{out} + d_{t,r_s}^{in})\frac{f(h)+g(r_s)-f(t)}{d(h,r_s,t)}\\
            &=\gamma N\sum_{(h,r_s,t)\in S}\left[\frac{f(h)+g(r_s)-f(t)}{d(h,r_s,t)}+\frac{f(t)+g(r_s)-f(h)}{d(t,r_s,h)}\right]\\&\quad-\frac{\gamma}{2}\sum_{h,t\in \mathscr{E}}\left[(d_{h,r_s}^{out}+d_{t,r_s}^{in})\frac{f(h)+g(r_s)-f(t)}{d(h,r_s,t)} +(d_{t,r_s}^{out}+d_{h,r_s}^{in})\frac{f(t)+g(r_s)-f(h)}{d(t,r_s,h)}\right]\\
            &=\gamma N\sum_{(h,r_s,t)\in S}\left[\frac{f(h)+g(r_s)-f(t)}{d(h,r_s,t)}+\frac{f(t)+g(r_s)-f(h)}{d(t,r_s,h)}\right]\\&\quad-\frac{\gamma}{2}\sum_{h,t\in \mathscr{E}}(d_{h,r_s}^{out}+d_{t,r_s}^{in})\left[\frac{f(h)+g(r_s)-f(t)}{d(h,r_s,t)}+\frac{f(t)+g(r_s)-f(h)}{d(t,r_s,h)}\right],
        \end{aligned}
    \end{equation}
    where $\gamma=\frac{2|S|}{M(N-1)}$. If $g(r_s)=0$, we can derive that $d(h,r_s,t)=\|h-t\|_2=d(t,r_s,h)$ and that $\frac{\partial \mathcal{L}}{\partial g(r_s)}=0$. 
    
    We further calculate the Hessian matrix $\mathcal{H}$ in the following:
    \begin{equation}
    \begin{aligned}
        \mathcal{H}&=\frac{\partial^2 \mathcal{L}}{\partial g(r_s)^2}\\
            &=\gamma N\sum_{(h,r_s,t)\in S}\frac{\partial^2 d(h,r_s,t)}{\partial g(r_s)^2}-\frac{\gamma} {2}\sum_{h,t\in \mathscr{E}}(d_{h,r_s}^{out} + d_{t,r_s}^{in})\frac{\partial^2 d(h,r_s,t)}{\partial g(r_s)^2}\\
            &=\gamma N\sum_{(h,r_s,t)\in S}\frac{d^2(h,r_s,t)I-[f(t)+g(r_s)-f(h)][f(t)+g(r_s)-f(h)]^T}{d^3(h,r_s,t)}\\
            &\quad -\frac{\gamma}{2}\sum_{h,t\in \mathscr{E}}(d_{h,r_s}^{out} + d_{t,r_s}^{in})\frac{d^2(h,r_s,t)I-[f(t)+g(r_s)-f(h)][f(t)+g(r_s)-f(h)]^T}{d^3(h,r_s,t)}.
    \end{aligned}
    \end{equation}
    For an arbitrary vector $x\in \mathbb{R}^d$ of unit length ($\|x\|_2=1$), we show that:
    \begin{equation}
    \label{eq:hessian}
    \begin{aligned}
        x^T\mathcal{H}x&=\gamma N\sum_{(h,r_s,t)\in S}\frac{d^2(h,r_s,t)x^Tx-(x^T[f(t)+g(r_s)-f(h)])^2}{d^3(h,r_s,t)}\\
        &\quad -\frac{\gamma}{2}\sum_{h,t\in \mathscr{E}}(d_{h,r_s}^{out} + d_{t,r_s}^{in})\frac{d^2(h,r_s,t)x^Tx-(x^T[f(t)+g(r_s)-f(h)])^2}{d^3(h,r_s,t)}\\
        &=\gamma N\sum_{(h,r_s,t)\in S}\frac{1-s^2(x, f(t)+g(r_s)-f(h))}{d(h,r_s,t)}\\
        &\quad -\frac{\gamma}{2}\sum_{h,t\in \mathscr{E}}(d_{h,r_s}^{out} + d_{t,r_s}^{in})\frac{1-s^2(x, f(t)+g(r_s)-f(h))}{d(h,r_s,t)}\\
        &=\frac{\gamma N}{2}\sum_{(h,r_s,t)\in S}\frac{1-\cos 2\alpha_{(h,r_s,t)}}{d(h,r_s,t)} -\frac{\gamma}{4}\sum_{h,t\in \mathscr{E}}(d_{h,r_s}^{out} + d_{t,r_s}^{in})\frac{1-\cos 2\alpha_{(h,r_s,t)}}{d(h,r_s,t)}\\
        &\ge\frac{\gamma N}{2}\sum_{(h,r_s,t)\in S}\frac{1}{d(h,r_s,t)} -\frac{\gamma}{4}\sum_{h,t\in \mathscr{E}}(d_{h,r_s}^{out} + d_{t,r_s}^{in})\frac{1}{d(h,r_s,t)}-\frac{3}{4}\gamma N\epsilon\\
        &=\frac{\gamma}{4}\left[\sum_{(h,r_s,t)\in S}\left[\sum_{\tilde{t}\in \mathscr{E},(h,r_s,\tilde{t})\notin S}\left(\frac{1}{d(h,r_s,t)}-\frac{1}{d(h,r_s,\tilde{t})}\right)\right.\right.\\
        &\qquad \qquad \qquad \left.\left.+\sum_{\tilde{h}\in \mathscr{E}, (\tilde{h},r_s,t)\notin S}\left(\frac{1}{d(h,r_s,t)}-\frac{1}{d(\tilde{h},r_s,t)}\right)\right]\right]-\frac{3}{4}\gamma N\epsilon\\
        &\ge \frac{\gamma}{4}(2N-d^{out}_{h,r_s}-d^{in}_{t,r_s})\epsilon-\frac{3}{4}\gamma N\epsilon\\
        &\ge \frac{\gamma N|S|\epsilon}{4}-\frac{3}{4}\gamma N\epsilon\\
        &\ge 0,
    \end{aligned}
    \end{equation}
    where $s(\cdot,\cdot)$ is cosine similarity. Eq. \ref{eq:hessian} shows that $\frac{\partial^2 \mathcal{L}}{\partial g(r_s)^2}$ is positive definite for any $g(r_s)\in \mathbb{R}^d$. Therefore, $g(r_s)=0$ is the minimum point of $\mathcal{L}$.

    Finally we derive the following based on Assumption \ref{assump:delta}:
    \begin{equation}
    \begin{aligned}
        \mathcal{L}_{kge}(h,r_s,t)&= \mathbb{E}_{\tilde{t}\sim \mathscr{E}\backslash t} \left[ (\|f(h)+g(r_s)-f(t)\|_2-\|f(h)+g(r_s)-f(\tilde{t})\|_2+\Delta) \right] \\
        &\quad +\mathbb{E}_{\tilde{h}\sim \mathscr{E}\backslash h} \left[ \|f(h)+g(r_s)-f(t)\|_2-\|f(\tilde{h})+g(r_s)-f(t)\|_2+\Delta \right] \\
        &\propto 2\|f(h)-f(t)\|_2-\mathbb{E}_{\tilde{t}\sim \mathscr{E}\backslash t} \|f(h)-f(\tilde{t})\|_2  -\mathbb{E}_{\tilde{h}\sim \mathscr{E}\backslash h}  \|f(\tilde{h})-f(t)\|_2.
    \end{aligned}
    \end{equation}
\end{proof}

\begin{definition}
    Functionally similar molecules. Assume that $h,t\in \mathscr{E}$ are two molecular entities. $h$ and $t$ are functionally similar if there exists some $o\in \mathscr{E}$ and $r\in \mathscr{R}$ that satisfies:
    $(h,r,o)\in KG, (t,r,o)\in KG$ or $(o,r,h)\in KG, (o,r,t)\in KG$. We define: 
    
    \begin{equation}
    \begin{aligned}
        &\mathcal{I}_1=\{(h,r,o), (t,r,o)|(h,r,o)\in KG,(t,r,o)\in KG\},\\ &\mathcal{I}_2=\{(o,r,h),(o,r,t)|(o,r,h)\in KG,(o,r,t)\in KG\},
    \end{aligned}
    \end{equation}
    
    and $\mathcal{I}=\mathcal{I}_1\cup \mathcal{I}_2$. We further assume that $|\mathcal{I}|\ll n$, indicating there are not too many intermediate entities connecting $h$ and $t$, which is common among biomedical knowledge bases. 
\end{definition}

\begin{lemma}
\label{lem:app_function}
For functionally similar molecules $h$ and $t$, the following holds:
\begin{equation}
    \|f(h)-f(t)\|\le \alpha \mathbb{E}_{(e_1,r,e_2)\sim \mathcal{I}}\left[ \mathcal{L}_{kge}(e_1,r,e_2) \right] + C,
\end{equation}
where $\alpha\approx 1, C\approx 0$ are constants.
\end{lemma}
\begin{proof}
    Following \citeApp{_qiu2018revisiting}, we rewrite $\mathcal{L}_{kge}$ as follows based on Eq. \ref{equ:neg_tail} and Eq. \ref{equ:neg_head}:
    \begin{equation}
    \label{equ:kge_rewrite}
        \begin{aligned}
        \mathcal{L}'&=\mathbb{E}_{(e_1,r,e_2)\sim \mathcal{I}} \left[ \mathcal{L}_{kge}(e_1,r,e_2) \right] \\&=\frac{1}{|\mathcal{I}|}\sum_{e_1,e_2\in \mathscr{E}, r\in \mathscr{R}}\frac{(2N\mathcal{X}_{e_1,r,e_2}-d_{e_1,r}^{out}-d_{e_2,r}^{in})d(e_1,r,e_2)}{N-1}+2\Delta,
        \end{aligned}
    \end{equation}
    where $\mathcal{X}_{e_1,r,e_2}=1$ indicates $(e_1,r,e_2)\in \mathcal{I}$ and $\mathcal{X}_{e_1,r,e_2}=0$ indicates $(e_1,r,e_2)\notin \mathcal{I}$.
    Further, the following inequalities hold:
    \begin{equation}
    \label{equ:degree}
        d_{e_1,r}^{out}=\sum_{e_2\in \mathscr{E}}\mathcal{X}_{e_1,r,e_2}\le |\mathcal{I}|,d_{e_2,r}^{in}=\sum_{e_1\in \mathscr{E}}\mathcal{X}_{e_1,r,e_2}\le |\mathcal{I}|.
    \end{equation}

    Based on Assumption. \ref{assump:delta} we have $d(e_1,r,e_2)\le \eta+\Delta$ where $\eta=\min_{(e_1,r,e_2)\in KG}[d(e_1,r,e_2)]$, and we assume that $\eta\approx 0$.
    
    Based on Eq. \ref{equ:kge_rewrite} and Eq. \ref{equ:degree}, we have:
    \begin{equation}
    \begin{aligned}
        \sum_{(e_1,r,e_2)\in \mathcal{I}}d(e_1,r,e_2)&\le \frac{1}{2(N-|\mathcal{I}|)}\sum_{(e_1,r,e_2)\in \mathcal{I}}(2N-d_{e_1,r}^{out}-d_{e_2,r}^{in})d(e_1,r,e_2)\\
        &=\frac{1}{2(N-|\mathcal{I}|)}\left[|\mathcal{I}|(N-1)(\mathcal{L}'-2\Delta)+\sum_{(e_1,r,e_2)\notin \mathcal{I}}(d_{e_1,r}^{out}+d_{e_2,r}^{in})d(e_1,r,e_2)\right]\\
        &\le \frac{1}{2(N-|\mathcal{I}|)}\left[|\mathcal{I}|(N-1)(\mathcal{L}'-2\Delta)+2|\mathcal{I}|N(\Delta+\eta)\right]\\
        &=\frac{|\mathcal{I}|(N-1)}{2(N-|\mathcal{I}|)}\mathcal{L}'+\frac{|\mathcal{I}|(N-1)\Delta+|\mathcal{I}|N(\Delta+\eta)}{N-|\mathcal{I}|}\\
        &=\frac{|\mathcal{I}|(N-1)}{2(N-|\mathcal{I}|)}\mathcal{L}'+\frac{|\mathcal{I}|(\Delta+N\eta)}{N-|\mathcal{I}|}
    \end{aligned}
    \end{equation}
    Then we have:
    \begin{equation}
    \begin{aligned}
        \|f(h)-f(t)\|&\le \min\left\{\min_{(h, r, o)\in \mathcal{I}_1} \{d(h,r,o)+d(t,r,o)\},\min_{(o, r, h)\in \mathcal{I}_2}\{d(o,r,h)+d(o,r,t)\}\right\}\\
        &\le \min\left\{\frac{2}{|\mathcal{I}_1|}\sum_{(e_1,r,e_2)\in \mathcal{I}_1}d(e_1,r,e_2),\frac{2}{|\mathcal{I}_2|}\sum_{(e_1,r,e_2)\in \mathcal{I}_2}d(e_1,r,e_2)\right\}\\
        &\le \frac{2}{|\mathcal{I}|}\sum_{(e_1,r,e_2)\in \mathcal{I}}d(e_1,r,e_2)\\
        &\le \frac{N-1}{N-|\mathcal{I}|}\mathcal{L}'+\frac{2(\Delta+N\eta)}{N-|\mathcal{I}|}\\
        &=\alpha \mathcal{L}'+C
    \end{aligned}
    \end{equation}
Since $|\mathcal{I}|\ll N$ and $\eta\approx 0$, we derive that $\alpha\approx 1$ and $C\approx 0$.
\end{proof}
\newpage

\subsection{Analysis of knowledge graph embedding}
\begin{figure*} [htpb]
\centering
    \includegraphics[width=\linewidth]{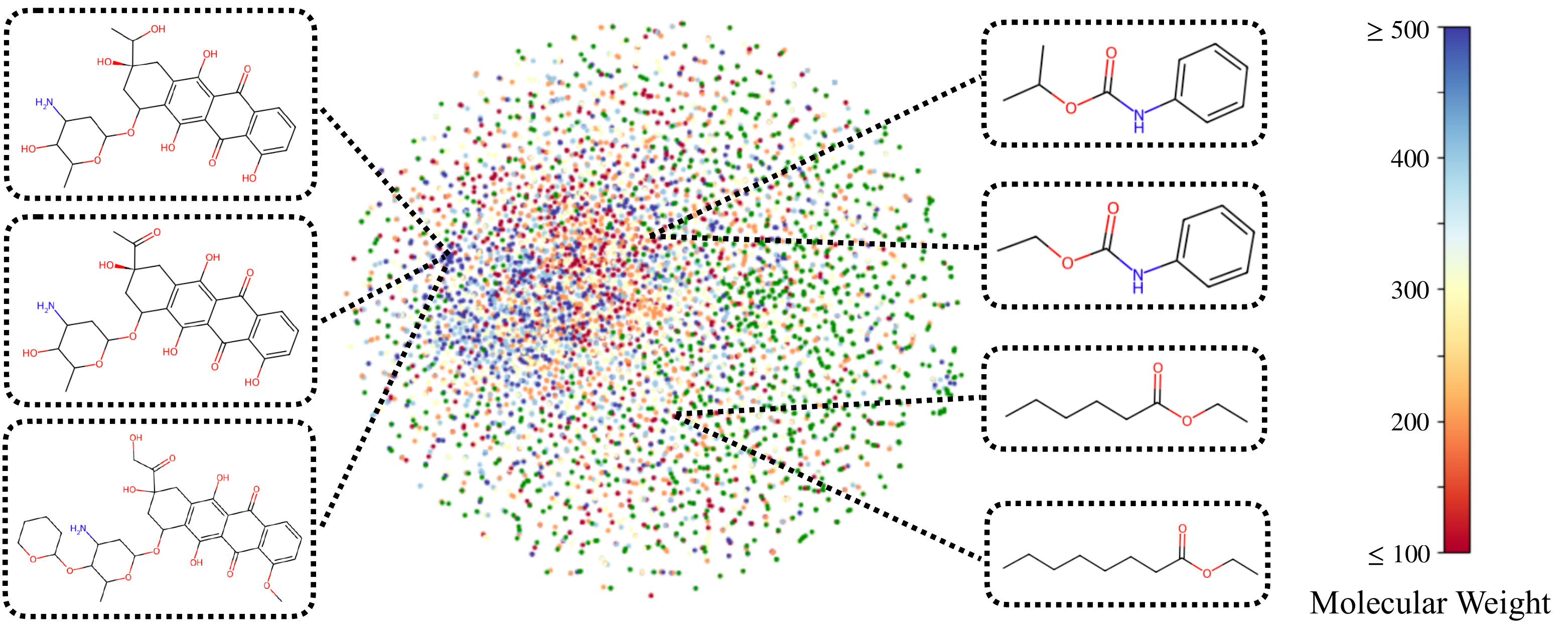}
\caption{Visualization of knowledge graph embeddings. Green dots represent entities that are not molecules. Other dots are colored based on molecular weight. We also present molecules that are structurally similar.}
\label{fig:kge}
\end{figure*}

\begin{table}[htpb]
\centering
\caption{Average distance between different molecules}
\label{tab:dist}
\begin{tabular}{lccc}
\toprule
Molecules     & structurally similar & functionally similar & random \\
\midrule
Avg. distance & 1.235                & 1.287                & 1.410  \\
\bottomrule
\end{tabular}

\end{table}

In this section, we present additional analysis of knowledge graph embeddings. In Fig. \ref{fig:kge} we illustrate the embeddings of MolFM knowledge encoder for 5,000 randomly sampled entities from the knowledge graph. These embeddings are then visualized using TSNE \citeApp{_van2008visualizing}, with molecules being color-coded based on their molecular weights. We also include randomly selected molecules that exhibit similar molecular structures or functions. Notably, Fig. \ref{fig:kge} demonstrates that the learned knowledge features show distinct clustering trends for structurally or functionally similar molecules. For instance, all three molecules on the left of the figure contain 4-amino-5-hydroxy-6-methyloxan-2-yl groups, and their pairwise Morgan fingerprint similarity \citeApp{_rogers2010extended} is no less than 0.78. 

Furthermore, we calculate the average distance between structurally similar molecules, functionally similar molecules and random molecules in Tab. \ref{tab:dist}. Though the distance between molecules sharing similar structures or functions are not close to 0, they display a significant margin compared to randomly selected molecules. In our experiments, we set a relatively small $\Delta=0.2$ to stabilize training, and the gradient is clipped to zero if the margin between positive samples and negative samples exceeds $\Delta$, which prohibits further optimization. However, it's still worth noting that KGE substantially brings structurally or functionally similar molecules closer while pushing dissimilar molecules apart.

\subsection{Pre-training dataset and knowledge graph details\label{sec:app_data}}
We utilize the same molecule-text pairs as introduced by \citeApp{_su2022molecular}. This dataset contains 15,613 molecules collected from PubChem \citeApp{_kim2016pubchem}, a comprehensive database of chemical substances and their biological activities as well as 37M paragraphs from S2ORC \citeApp{_lo2020s2orc}, a versatile corpus for text mining in scientific papers. \citeApp{_su2022molecular} utilizes simple rules such as using molecular names as queries to obtain molecule-text pairs. Then, we build a knowledge graph for the 15,613 molecules and more with the following steps: 

\textbf{Aligning entities in different databases.} The knowledge graph focuses on drugs (molecules), proteins (targets), diseases and other biomedical entities. We collect additional molecules from DrugBank \citeApp{_wishart2018drugbank}, a public database containing structured drug information, and perform duplicate elimination by comparing the isomeric SMILES strings to the 15,613 molecules in our pre-training data. Proteins are identified using Uniprot \citeApp{_uniprot2019uniprot}, a widely used protein database. We identify diseases and other entities using MeSH (Medical Subject Headings) \citeApp{_lipscomb2000medical}, a standard vocabulary thesaurus maintained by U.S. National Library of Medicine. 

\textbf{Building connections between entities.} The knowledge graph consists of relations including drug-target interaction, drug-drug similarity relationship, drug-drug interaction, and drug-disease association. We build these connections in the following:

For drug-target interactions, we collect drug targets, drug enzymes, drug carriers, and drug transporters from DrugBank. Furthermore, we incorporate BindingDB \citeApp{_gilson2016bindingdb}, a public database of biomolecular interactions based on binding affinities. We compare the isomeric SMILES of our molecules with the BindingDB compounds, and extract their protein targets with binding affinity values $Ki\le 10 nM$.

For drug-drug similarity relationships, we leverage MHFP \citeApp{_probst2018probabilistic}, an efficient molecular fingerprint to find $k$-nearest neighbors from all the molecules in our knowledge graph (we use $k=10$ in the study). We further compare the RDKit fingerprint similarity \citeApp{_schneider2015get} between the molecule with the $k$ candidates, and use a threshold of 0.8 to build drug-drug similarity relations. In cases where none of the candidates satisfies the threshold, we further lower the threshold to 0.6 to ensure connectivity.

For drug-drug interactions, we adopt relationships from DrugBank, and further categorize them into 12 classes based on the patterns of their textual description, including \textit{increased activities}, \textit{decreased activities}, \textit{increase risk/severity of adverse effect}, \textit{decrease risk/severity of adverse effect}, \textit{increased metabolism}, \textit{decreased metabolism}, \textit{increase of therapeutic efficacy}, \textit{decrease of therapeutic efficacy}, \textit{increased excretion rate}, \textit{decreased excretion rate}, \textit{increased serum concentration}, \textit{decreased serum concentration}.

For relationships between drugs, diseases and other entities, we collect data from the online platform of FORUM \citeApp{_delmas2021building}, a knowledge base that supports queries for PubChem molecules. We select the most trustworthy associations with $q\_value<10^{-6}$.

The overall statistics of our knowledge graph are presented in Tab. \ref{tab:kg_stat}.
\begin{table}[htpb]\small
\centering
\setlength\tabcolsep{4pt}
\caption{Statistics of entities and relations of our knowledge graph. ddi denotes drug-drug interaction.}
\label{tab:kg_stat}
\begin{tabular}{llll}
\toprule
\multicolumn{2}{l}{Entities}                           & ddi: increased metabolism             & 110,958    \\ \cline{1-2}
molecules                                & 29,043     & ddi: decreased metabolism             & 288,010    \\
diseases                                 & 19,655     & ddi: increase of therapeutic efficacy & 46,492     \\
proteins                                 & 403         & ddi: decrease of therapeutic efficacy & 211,108    \\
All                                      & 49,111     & ddi: increased excretion rate         & 56,768     \\ \cline{1-2}
Relations                                &             & ddi: decreased excretion rate         & 390,120    \\ \cline{1-2}
drug-protein interaction                 & 23,870     & ddi: increased serum concentration    & 79,536     \\
ddi: increased activities          & 294,738    & ddi: decreased serum concentration    & 25,048     \\
ddi: decreased activities                     & 82,712     & drug-drug similarity             & 95,804     \\
ddi: increase risk/severity of adverse effect & 1,044,749 & drug-disease                     & 499,745    \\
ddi: decrease risk/severity of adverse effect & 880         & All                              & 3,253,238 \\
\bottomrule
\end{tabular}
\end{table}

\subsection{Downstream task details\label{sec:app_downstream}}
Here we provide the implementation details for fine-tuning MolFM and other baseline models. For all fine-tuning experiments, we use Adam optimizer with a weight decay of $10^{-5}$ and select a learning rate from $\{10^{-4}, 3\times 10^{-4},10^{-3}\}$. We run experiments for either $100$ or $200$ epochs with 3 different random seeds. We employ early-stopping with a patience of $20$ epochs. 

\textbf{Cross-modal retrieval.} We evaluate our model on the modified PCdes \citeApp{_zeng2022deep} dataset. The original PCdes is collected from PubChem and consists of 15K molecules. We remove 8 molecules whose SMILES strings could not be transformed into a 2D graph by RDKit, and filter out 3,880 molecules that have 
appeared in our pre-training dataset to prevent information leakage. We adopt Scaffold split \citeApp{_bemis1996properties} instead of random split to evaluate the generalization capability of retrieval models with a train/validation/test ratio of 7:1:2. We conduct both paragraph-level and sentence-level cross-modal retrieval. In paragraph-level retrieval, we use the whole description for the molecule as text input. In sentence-level retrieval, we randomly pick one sentence for each molecule as text input. During fine-tuning, we optimize max of hinge loss between the cosine similarity of structural and textual representations within a minibatch of size 32. For SciBERT \citeApp{_beltagy2019scibert}, KV-PLM \citeApp{_zeng2022deep} and KV-PLM* \citeApp{_zeng2022deep}, we use the language model to simultaneously encode 1D SMILES strings and texts. As for the GraphMVP \citeApp{_graphmvp} baseline, we use GraphMVP to encode 2D molecular graphs and employ SciBERT to encode texts. 

\textbf{Molecule captioning.} We utilize the ChEBI-20 \citeApp{_edwards2021text2mol} dataset with 33,010 molecule-description pairs. We follow the original 8:1:1 train/validation/test split. Evaluation metrics include BLEU \citeApp{_papineni2002bleu}, ROUGE \citeApp{_lin2004rouge}, METEOR \citeApp{_banerjee2005meteor} and Text2Mol score \citeApp{_edwards2021text2mol}. GraphMVP shares the same architecture as MolFM, where atom features are concatenated with the outputs of the MolT5 \citeApp{_edwards2022translation} encoder. The concatenation result is then fed into the MolT5 decoder to generate molecular descriptions.

\textbf{Text-based molecule generation.} We conduct experiments on ChEBI-20 with the same split as molecule captioning. Evaluation metrics include BLEU, exact ratio (ratio of generated SMILES strings that are identical to the ground truth), valid ratio (ratio of generated SMILES strings that correspond to valid molecules), Levenshtein distance \citeApp{_miller2009advanced}, fingerprint Tanimoto similarity (we use MACCS fingerprint \citeApp{_durant2002reoptimization}, Morgan fingerprint \citeApp{_rogers2010extended}, RDKit fingerprint \citeApp{_schneider2015get}) and Text2Mol score. For SciBERT and MoMu, we feed the outputs of the 6th transformer layer into the MolT5 decoder to ensure that they contain the same amount of parameters as MolFM's text encoder.

\textbf{Molecular property prediction.} We adopt classification datasets in MoleculeNet, a widely used molecular property benchmark. Tab. \ref{tab:mp} provides a summary of the dataset statistics. We follow the same Scaffold split as \citeApp{_graphmvp} with a train/validation/test ratio of 8:1:1. To obtain knowledge inputs for each molecule in the dataset, we first compare the isomeric SMILES to molecules in the knowledge graph for an exact match. If there is no exact match, we select a molecule entity in our knowledge graph that has the highest RDKit fingerprint Tanimoto similarity. If the fingerprint similarity is not greater than 0.8, the knowledge input will be a \textit{"null"} entity with random embeddings. For additional text inputs for each molecule in the dataset, we compare the isomeric SMILES to molecules in ChEBI-20 to find an exact match and obtain the corresponding description. If there is no exact match, the text input will be \textit{"No description for the drug is available"}. During fine-tuning, we perform additional hyper-parameter search on the dropout ratio of MolFM's structure encoder from $\{0, 0.1, 0.3, 0.5\}$. 

\begin{table}[htpb]
\centering
\caption{Summary of molecular property prediction datasets. \# Molecules: number of molecules. \# Tasks: number of prediction objectives. \# Linked to KG: number of molecules that we obtain knowledge graph inputs. \# Linked to text: number of molecules that we obtain text inputs.}
\label{tab:mp}
\begin{tabular}{lllllllll}
\toprule
Dataset      & BBBP  & Tox21 & ToxCast & SIDER & ClinTox & MUV    & HIV    & BACE  \\
\midrule
\# Molecules & 2,039 & 7,831 & 8,597   & 1,427 & 1,478   & 93,807 & 41,127 & 1,513 \\
\# Tasks     & 1     & 12    & 617     & 27    & 2       & 17     & 1      & 1     \\
\midrule
\# Linked to KG & 1,605 & 6,328 & 6,892 & 1,140 & 1,151 & 9,006 & 8,131 & 232 \\
\# Linked to text & 599 & 2,537 & 2,538 & 599 & 585 & 219 & 426 & 3 \\
\bottomrule
\end{tabular}
\end{table}

\subsection{Additional experiments \label{sec:app_exp}}
Tab. \ref{cmr-para-s-t} and Tab. \ref{cmr-para-t-s} show the paragraph-level cross-modal retrieval results and error bars under fine-tuning setting. Tab. \ref{cmr-sent-s-t} and Tab. \ref{cmr-sent-t-s} show the sentence-level cross-modal retrieval results and error bars under zero-shot and fine-tuning settings. Tab. \ref{molcap1} and Tab. \ref{molcap2} present the molecule captioning results and error bars. Tab. \ref{t2mgen1} and Tab. \ref{t2mgen2} display the text-based molecule generation results and error bars.

In addition, we conduct ablation studies on the number of neighbors $N$. We pre-train MolFM with different choices of $N$ and evaluate the zero-shot paragraph-level cross-modal retrieval performance, as shown in Fig. \ref{fig:ablation_neighbor}. We observe that when $N\le 4$, aggregating information from more neighbors slightly improves the retrieval performance. However, when $N>4$, increasing $N$ has little impact on our model, which can be attributed to two reasons. Firstly, the sparsity of our knowledge graph results in only a few entities being connected to more than 4 neighbors. Secondly, the interaction relationships between molecules and other entities may exhibit certain patterns or dependencies. Hence, including additional neighbors beyond a certain point may introduce redundant information that does not provide substantial benefits to the representation learning of our model.

\begin{table}[htpb!]
\centering
\caption{Fine-tuned paragraph-level structure-to-text (S-T) retrieval results on the test split of PCdes. }
\label{cmr-para-s-t}
\begin{tabular}{lcccc}
\toprule
Model    & MRR & R@1 & R@5 & R@10 \\
\midrule
SciBERT  & 24.98$_{\pm 0.88}$ & 16.32$_{\pm 0.92}$ & 33.91$_{\pm 0.88}$ & 42.64$_{\pm 1.88}$      \\
KV-PLM   & 27.41$_{\pm 0.80}$ & 18.35$_{\pm 0.70}$ & 37.15$_{\pm 1.19}$ & 45.43$_{\pm 0.79}$      \\
KV-PLM*  & 29.15$_{\pm 0.47}$ & 20.60$_{\pm 0.53}$ & 37.87$_{\pm 0.65}$ & 45.74$_{\pm 0.56}$      \\
GraphMVP & 31.57$_{\pm 0.64}$ & 23.26$_{\pm 0.67}$ & 40.21$_{\pm 0.41}$ & 47.39$_{\pm 0.63}$      \\
MoMu     & 34.29$_{\pm 0.69}$ & 24.47$_{\pm 0.64}$ & 45.38$_{\pm 1.25}$  & 53.84$_{\pm 0.83}$      \\
MolFM    & \textbf{39.56$_{\pm 0.64}$} & \textbf{29.76$_{\pm 0.70}$}    & \textbf{50.53$_{\pm 0.38}$}    & \textbf{58.63$_{\pm 0.26}$}  \\
\bottomrule
\end{tabular}
\end{table}

\begin{table}[htpb!]
\centering
\caption{Fine-tuned paragraph-level text-to-structure (T-S) retrieval results on the test split of PCdes. }
\label{cmr-para-t-s}
\begin{tabular}{lcccc}
\toprule
Model    & MRR & R@1 & R@5 & R@10 \\
\midrule
SciBERT  & 23.92$_{\pm 0.80}$ & 14.97$_{\pm 0.79}$ & 34.05$_{\pm 1.03}$ & 41.74$_{\pm 1.88}$      \\
KV-PLM   & 25.97$_{\pm 1.04}$ & 16.55$_{\pm 1.25}$ & 35.85$_{\pm 1.15}$ & 44.75$_{\pm 0.86}$      \\
KV-PLM*  & 28.12$_{\pm 0.49}$ & 19.29$_{\pm 0.45}$ & 37.33$_{\pm 0.53}$ & 45.29$_{\pm 0.26}$      \\
GraphMVP & 30.93$_{\pm 0.40}$ & 21.94$_{\pm 0.52}$ & 40.28$_{\pm 0.25}$ & 47.90$_{\pm 0.39}$      \\
MoMu     & 34.53$_{\pm 1.54}$ & 24.87$_{\pm 1.55}$ & 44.93$_{\pm 1.51}$ & 54.25$_{\pm 1.27}$      \\
MolFM    & \textbf{39.34$_{\pm 0.70}$} & \textbf{29.39$_{\pm 0.81}$}    & \textbf{50.26$_{\pm 0.65}$}    & \textbf{58.49$_{\pm 0.98}$}  \\
\bottomrule
\end{tabular}
\end{table}

\begin{table}[htpb!]
\centering
\caption{Sentence-level structure-to-text (S-T) retrieval results on the test split of PCdes. }
\label{cmr-sent-t-s}
\begin{tabular}{llllll}
\toprule
Mode & Model    & MRR & R@1 & R@5 & R@10 \\
\midrule
\multirow{2}{*}{zero-shot} & MoMu & 5.95 & 3.05 & 7.24 & 10.97 \\
& MolFM & \textbf{12.54} & \textbf{8.00} & \textbf{16.10} & \textbf{21.23} \\
\midrule
\multirow{6}{*}{fine-tune} & SciBERT  & 12.27$_{\pm 0.27}$ & 6.59$_{\pm 0.21}$ & 17.26$_{\pm 0.27}$ & 23.16$_{\pm 0.41}$      \\
& KV-PLM  & 12.93$_{\pm 0.91}$ & 7.15$_{\pm 1.01}$ & 17.84$_{\pm 0.70}$ & 23.88$_{\pm 0.41}$      \\
& KV-PLM*   & 14.59$_{\pm 0.29}$ & 8.64$_{\pm 0.24}$ & 19.98$_{\pm 0.54}$ & 26.22$_{\pm 0.42}$      \\
& GraphMVP & 14.76$_{\pm 1.09}$ & 8.96$_{\pm 1.01}$ & 19.84$_{\pm 1.42}$ & 25.70$_{\pm 1.16}$      \\
& MoMu     & 19.91$_{\pm 0.66}$ & 12.98$_{\pm 0.81}$ & 26.66$_{\pm 0.81}$  & 33.64$_{\pm 0.66}$      \\
& MolFM    & \textbf{21.14$_{\pm 0.80}$} & \textbf{14.09$_{\pm 0.75}$}    & \textbf{28.18$_{\pm 0.82}$}    & \textbf{35.31$_{\pm 0.68}$}  \\
\bottomrule
\end{tabular}
\end{table}

\begin{table}[htpb!]
\centering
\caption{Sentence-level text-to-structure (T-S) retrieval results on the test split of PCdes. }
\label{cmr-sent-s-t}
\begin{tabular}{llllll}
\toprule
Mode & Model    & MRR & R@1 & R@5 & R@10 \\
\midrule
\multirow{2}{*}{zero-shot} & MoMu  & 6.18 & 3.01 & 7.73 & 12.37 \\
& MolFM & \textbf{13.48} & \textbf{8.23} & \textbf{17.76} & \textbf{22.98} \\
\midrule
\multirow{6}{*}{fine-tune} & SciBERT  & 11.79$_{\pm 0.42}$ & 6.25$_{\pm 0.45}$ & 16.41$_{\pm 0.40}$ & 22.46$_{\pm 0.23}$      \\
& KV-PLM  & 12.29$_{\pm 0.83}$ & 6.71$_{\pm 0.83}$ & 16.79$_{\pm 0.88}$ & 23.49$_{\pm 0.28}$      \\
& KV-PLM*   & 14.24$_{\pm 0.26}$ & 8.28$_{\pm 0.15}$ & 19.72$_{\pm 0.30}$ & 26.28$_{\pm 0.42}$      \\
& GraphMVP & 14.75$_{\pm 1.20}$ & 9.04$_{\pm 1.02}$ & 19.73$_{\pm 1.79}$ & 25.60$_{\pm 1.72}$      \\
& MoMu     & 20.10$_{\pm 1.07}$ & 13.23$_{\pm 1.09}$ & 26.81$_{\pm 1.32}$ & 33.76$_{\pm 1.11}$      \\
& MolFM    & \textbf{21.54$_{\pm 0.11}$} & \textbf{14.49$_{\pm 0.24}$}    & \textbf{28.46$_{\pm 0.46}$}    & \textbf{35.82$_{\pm 0.35}$}  \\
\bottomrule
\end{tabular}
\end{table}

\begin{table}[htpb]
\centering
\caption{BELU and ROUGE scores of molecule captioning on the test split of ChEBI-20. $^{\dag}$: These results are taken from \cite{_edwards2022translation}.}
\label{molcap1}
\setlength\tabcolsep{2pt}
\begin{tabular}{lllllll}
\toprule
Decoder     & Encoder     & BLEU-2 & BLEU-4 & ROUGE-1 & ROUGE-2 & ROUGE-L \\
\midrule
\multirow{4}{*}{MolT5-small} & MolT5-small$^{\dag}$   & 0.519  & 0.436  & 0.620   & \textbf{0.469}   & 0.563   \\
& MoMu  & 0.532$_{\pm 0.001}$  & 0.445$_{\pm 0.000}$  & 0.621$_{\pm 0.000}$   & \textbf{0.469$_{\pm 0.000}$}   & \textbf{0.564$_{\pm 0.001}$}   \\
& GraphMVP  & 0.540$_{\pm 0.002}$  & 0.449$_{\pm 0.001}$  & 0.619$_{\pm 0.002}$   & 0.465$_{\pm 0.002}$   & 0.560$_{\pm 0.001}$    \\
& MolFM       & \textbf{0.542$_{\pm 0.002}$}  & \textbf{0.452$_{\pm 0.001}$}  & \textbf{0.623$_{\pm 0.001}$}   & \textbf{0.469$_{\pm 0.001}$}   & 0.562$_{\pm 0.002}$    \\
\midrule
\multirow{4}{*}{MolT5-base} & MolT5-base$^{\dag}$   & 0.540  & 0.457  & 0.634   & 0.485   & 0.578 \\
& MoMu   & 0.549$_{\pm 0.000}$  & 0.462$_{\pm 0.000}$  & 0.630$_{\pm 0.001}$   & 0.479$_{\pm 0.000}$    & 0.575$_{\pm 0.000}$  \\
& GraphMVP    & 0.577$_{\pm 0.003}$  & 0.491$_{\pm 0.002}$  & 0.651$_{\pm 0.002}$  & 0.505$_{\pm 0.002}$    & 0.592$_{\pm 0.002}$  \\
& MolFM    & \textbf{0.585$_{\pm 0.002}$}  & \textbf{0.498$_{\pm 0.001}$}  & \textbf{0.653$_{\pm 0.002}$}   & \textbf{0.508$_{\pm 0.001}$}   & \textbf{0.594$_{\pm 0.002}$} \\
\bottomrule
\end{tabular}
\end{table}

\begin{figure}[htpb!]
    \centering
    \begin{minipage}{0.55\linewidth}
    \centering
    \caption{MEATOR and Text2Mol scores of molecule captioning on the test split of ChEBI-20. $^{\dag}$: These results are taken from \cite{_edwards2022translation}.}
    \label{molcap2}
    \begin{tabular}{llll}
    \toprule
    Decoder     & Encoder     & METEOR & Text2Mol \\
    \midrule
    \multirow{4}{*}{\makecell[l]{MolT5\\-small}} & MolT5-small$^{\dag}$   & 0.551  & 0.540    \\
    & MoMu   & 0.557$_{\pm 0.001}$  & 0.543$_{\pm 0.001}$    \\
    & GraphMVP & 0.562$_{\pm 0.002}$  & 0.553$_{\pm 0.003}$    \\
    & MolFM  & \textbf{0.564$_{\pm 0.002}$}  & \textbf{0.557$_{\pm 0.002}$}    \\
    \midrule
    \multirow{4}{*}{\makecell[l]{MolT5\\-small}} & MolT5-base$^{\dag}$   & 0.569  & 0.547    \\
    & MoMu  & 0.576$_{\pm 0.001}$  & 0.558$_{\pm 0.000}$    \\
    & GraphMVP    & 0.599$_{\pm 0.003}$  & 0.570$_{\pm 0.002}$        \\
    & MolFM    & \textbf{0.607$_{\pm 0.002}$}  & \textbf{0.576$_{\pm 0.002}$}   \\
    \bottomrule
    \end{tabular}
    \end{minipage}
    \hfill
    \begin{minipage}{0.4\linewidth}
        \caption{Cross-modal retrieval results with different number of sampled neighbors $N$.}
        \includegraphics[width=\linewidth]{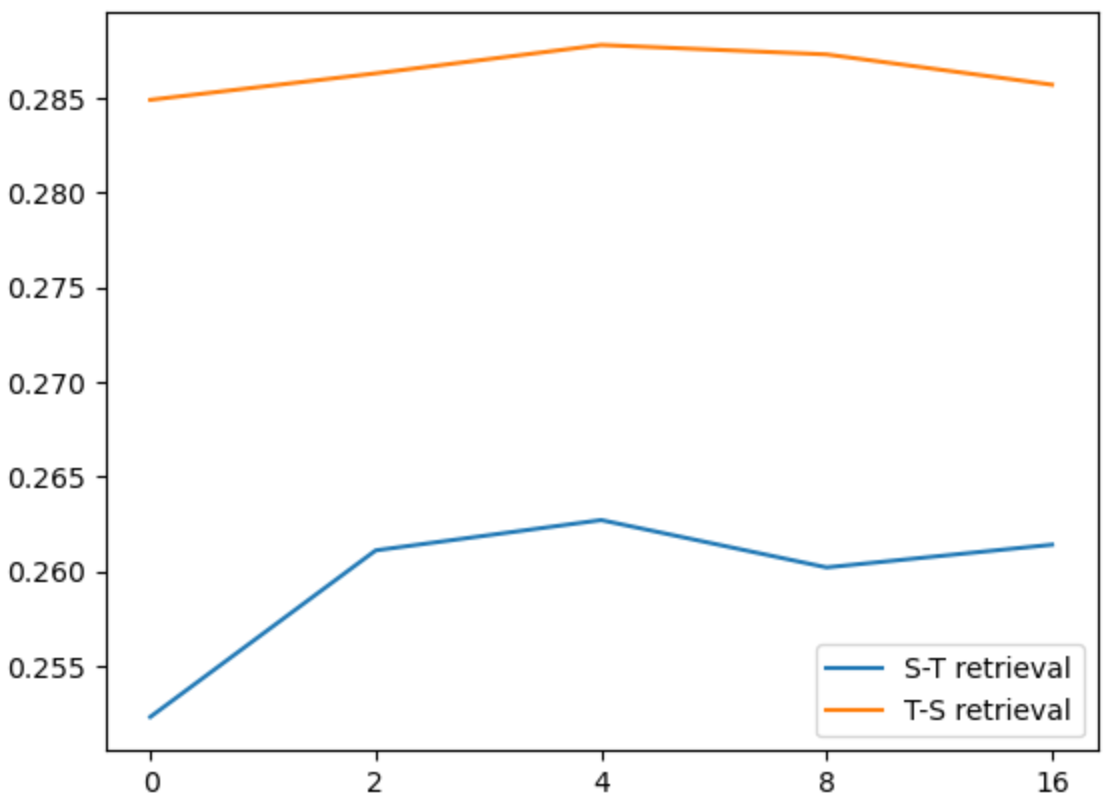}
        \label{fig:ablation_neighbor}
    \end{minipage}
    
\end{figure}
\begin{table}[htpb!]
\centering

\setlength\tabcolsep{2pt}

\end{table}

\begin{table}[htpb!]
\setlength\tabcolsep{3pt}
\centering
\caption{Text-based molecule generation results on the test split of ChEBI-20. $\uparrow$: The higher the better. $\downarrow$: The lower the better. $^{\dag}$: These results are taken from \cite{_edwards2022translation}.}
\label{t2mgen1}
\begin{tabular}{llllll}
\toprule
Decoder     &  Encoder                      & BLEU $\uparrow$  & Exact $\uparrow$ & Valid $\uparrow$ &  Levenshtein $\downarrow$ \\
\midrule
\multirow{4}{*}{MolT5-small} & MolT5-small  & 0.749 & 0.081 & 0.724 & 29.160    \\
& SciBERT   & 0.797$_{\pm 0.002}$ & 0.142$_{\pm 0.015}$ & 0.846$_{\pm 0.017}$ & 22.027$_{\pm 0.645}$    \\
& MoMu   & 0.800$_{\pm 0.003}$ & 0.150$_{\pm 0.017}$ & 0.858$_{\pm 0.011}$ & 21.446$_{\pm 0.733}$ \\
& MolFM   & \textbf{0.803$_{\pm 0.002}$} & \textbf{0.169$_{\pm 0.012}$} & \textbf{0.859$_{\pm 0.008}$} & \textbf{20.868$_{\pm 0.598}$}  \\
\midrule
\multirow{4}{*}{MolT5-base} & MolT5-base  & 0.779 & 0.082 & 0.786      & 25.188 \\
& SciBERT  & 0.812$_{\pm 0.002}$ & 0.179$_{\pm 0.011}$ & 0.852$_{\pm 0.014}$     & 21.192$_{\pm 0.612}$   \\
& MoMu  & 0.815$_{\pm 0.002}$ & 0.183$_{\pm 0.014}$ & 0.863$_{\pm 0.014}$     & 20.520$_{\pm 0.757}$        \\
& MolFM    & \textbf{0.822$_{\pm 0.002}$} & \textbf{0.210$_{\pm 0.013}$} & \textbf{0.892$_{\pm 0.012}$} & \textbf{19.445$_{\pm 0.745}$} \\
\bottomrule
\end{tabular}
\end{table}

\begin{table}[htpb!]
\setlength\tabcolsep{3pt}
\centering
\caption{Text-based molecule generation results on the test split of ChEBI-20. $\uparrow$: The higher the better. $\downarrow$: The lower the better. $^{\dag}$: These results are taken from \cite{_edwards2022translation}.}
\label{t2mgen2}
\begin{tabular}{llllll}
\toprule
Decoder     &  Encoder                      & MACCS FTS $\uparrow$ & RDKit FTS $\uparrow$ & Morgan FTS $\uparrow$ & Text2Mol $\uparrow$ \\
\midrule
\multirow{4}{*}{MolT5-small} & MolT5-small$^{\dag}$  & 0.780     & 0.653   & 0.601      & 0.533    \\
& SciBERT  & 0.818$_{\pm 0.005}$     & 0.695$_{\pm 0.009}$   & 0.639$_{\pm 0.016}$      & 0.561$_{\pm 0.007}$    \\
& MoMu   & 0.818$_{\pm 0.007}$     & 0.709$_{\pm 0.010}$   & 0.651$_{\pm 0.009}$      & 0.566$_{\pm 0.004}$    \\
& MolFM    & \textbf{0.834$_{\pm 0.006}$}     & \textbf{0.721$_{\pm 0.008}$}   & \textbf{0.662$_{\pm 0.011}$}      & \textbf{0.573$_{\pm 0.004}$}    \\
\midrule
\multirow{4}{*}{MolT5-base} & MolT5-base$^{\dag}$  & 0.787 & 0.661  & 0.601     & 0.543        \\
& SciBERT   & 0.844$_{\pm 0.008}$         & 0.733$_{\pm 0.011}$       & 0.678$_{\pm 0.012}$          & 0.575$_{\pm 0.005}$        \\
& MoMu  & 0.847$_{\pm 0.006}$         & 0.737$_{\pm 0.013}$       & 0.678$_{\pm 0.010}$          & 0.580$_{\pm 0.003}$        \\
& MolFM     & \textbf{0.854$_{\pm 0.005}$}     & \textbf{0.758$_{\pm 0.012}$}   & \textbf{0.697$_{\pm 0.009}$}      & \textbf{0.583$_{\pm 0.004}$}    \\
\bottomrule
\end{tabular}
\end{table}

\newpage
\subsection{Additional downstream task cases}

\subsubsection{Cross-modal retrieval \label{sec:app_case_mtr}}

Fig. \ref{fig:st_suppl} and Fig. \ref{fig:ts_suppl} show comparisons between MolFM and MoMu on structure-to-text retrieval and text-to-structure retrieval. We present the structure or text inputs, the top-3 retrieved results for two models, along with the prediction scores and whether the retrieved candidates hit the ground truth. 
\begin{figure*} [htpb]
\centering
    \includegraphics[width=0.95\linewidth]{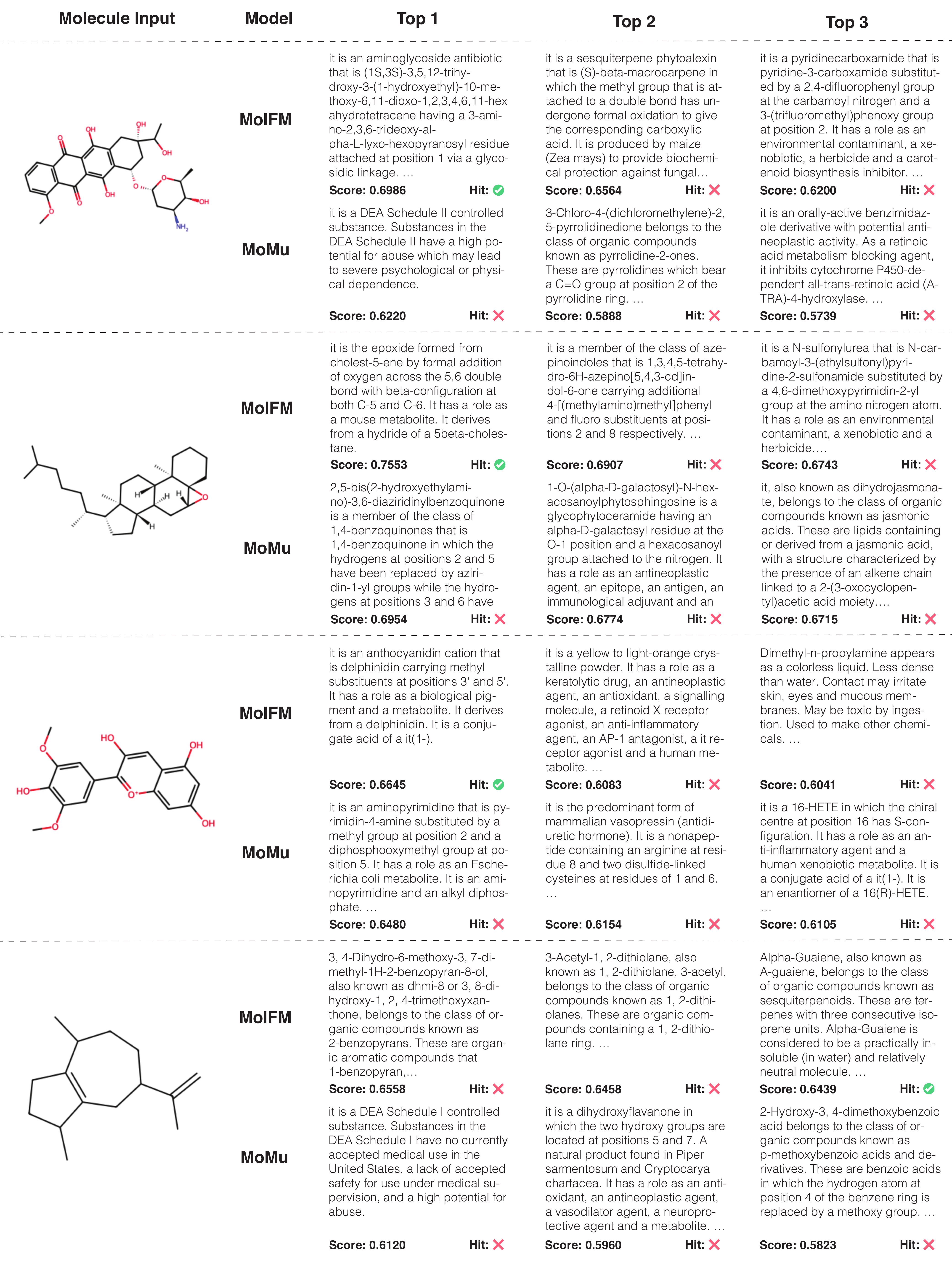}
\caption{Structure-to-text retrieval examples.}
\label{fig:st_suppl}
\end{figure*}

\newpage
\begin{figure*} [htpb]
\centering
    \includegraphics[width=\linewidth]{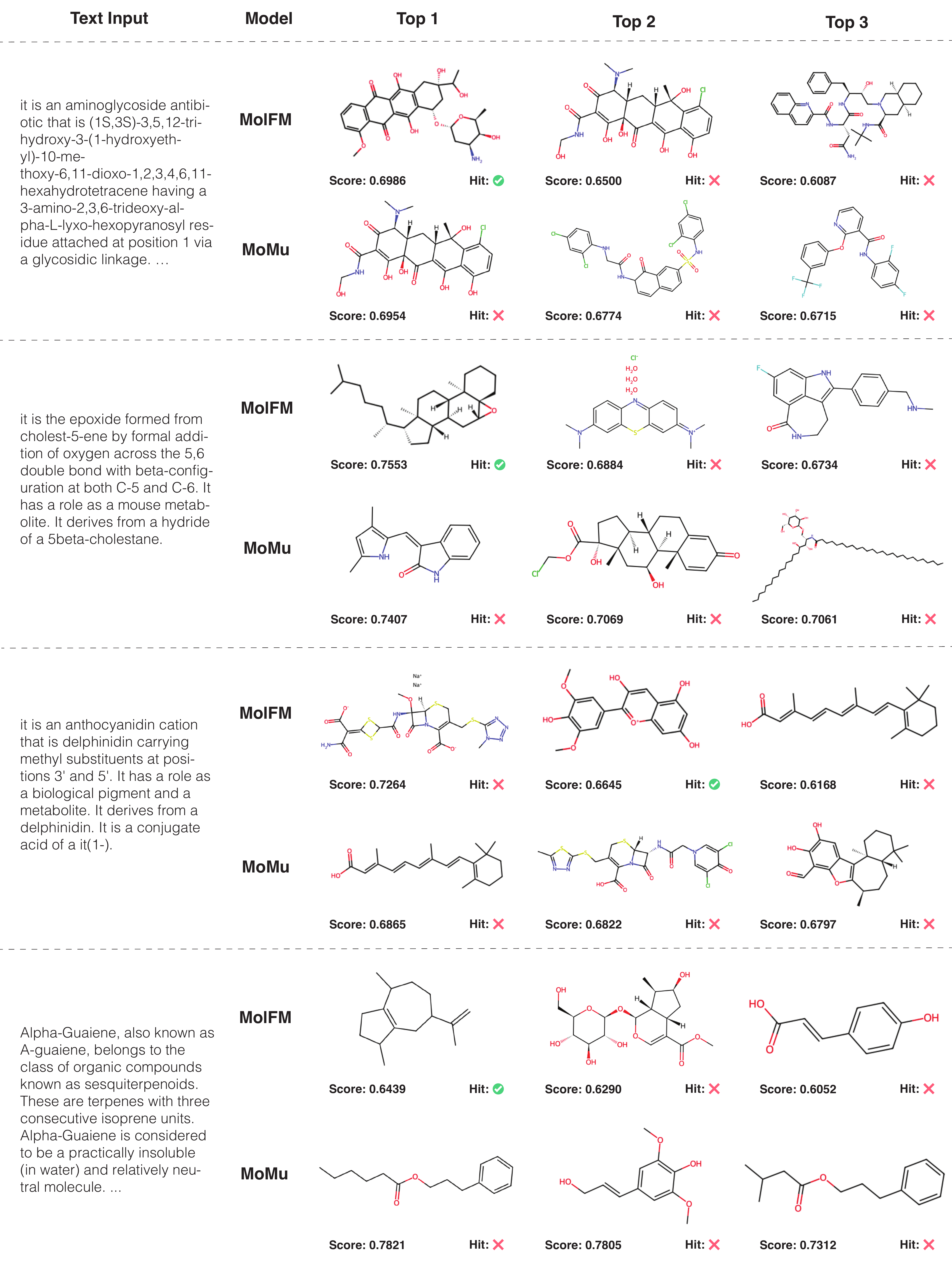}
\caption{Text-to-structure retrieval examples.}
\label{fig:ts_suppl}
\end{figure*}

\newpage
\subsubsection{Molecule captioning \label{sec:app_case_mp}}

In Fig. \ref{fig:molcap_suppl}, we illustrate diverse molecules as well as the molecule captioning results of different models. We highlight the text fragments where MolFM generates more accurate expressions that shares similar or exact semantics with the ground truth. However, such contents are missing or incorrect in the outputs of other models. 
\begin{figure*} [htpb]
\centering
    \includegraphics[width=\linewidth]{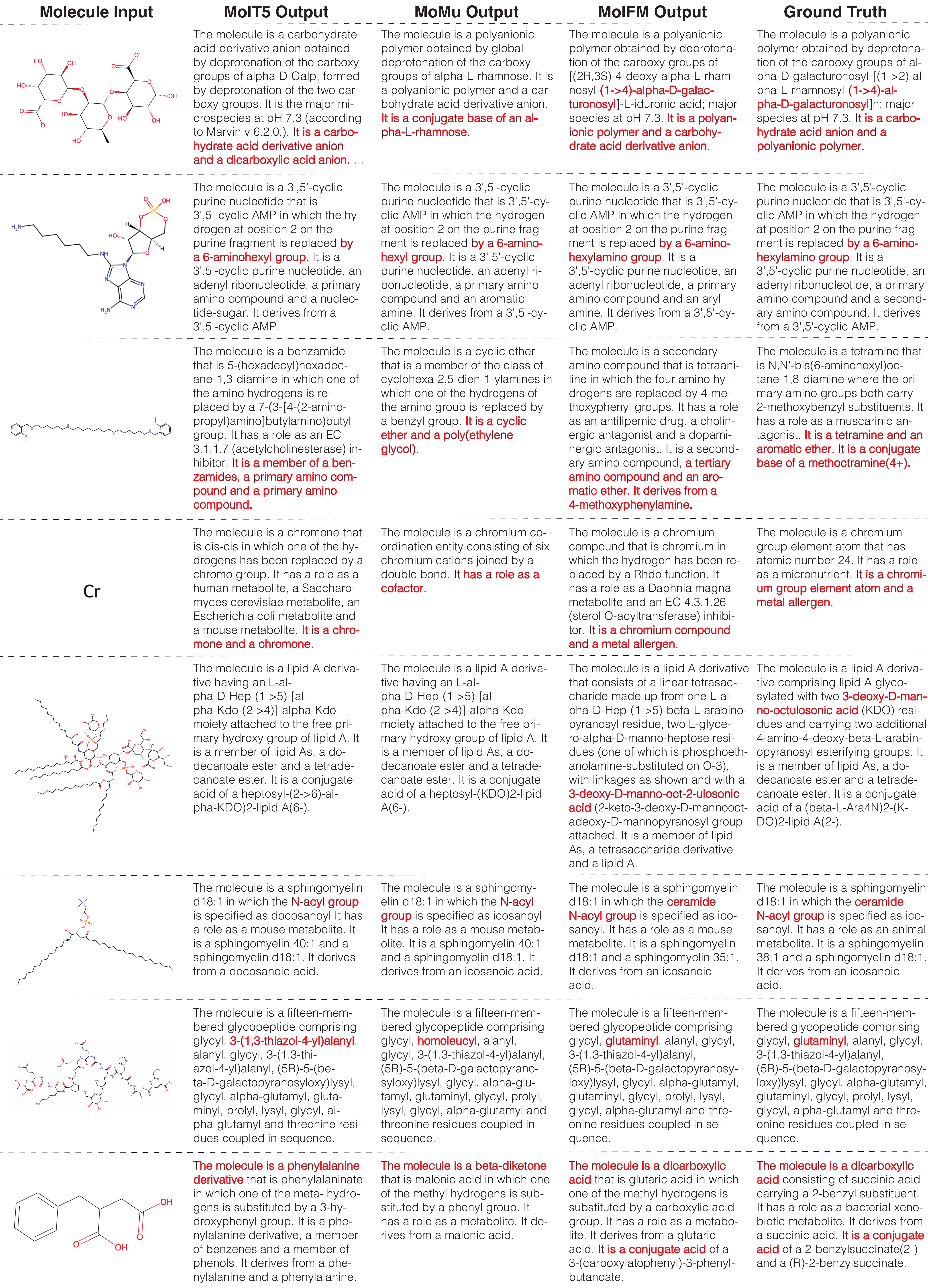}
\caption{Additional molecule captioning examples.}
\label{fig:molcap_suppl}
\end{figure*}

\subsubsection{Text-to-molecule generation \label{sec:app_case_t2m}}

Fig. \ref{fig:text2smi_suppl} shows text-to-molecule generation results of different models. We also calculate Morgan fingerprint Tanimoto similarity between the generated molecules and the ground truth.
\begin{figure*} [htpb!]
\centering
    \includegraphics[width=0.93\linewidth]{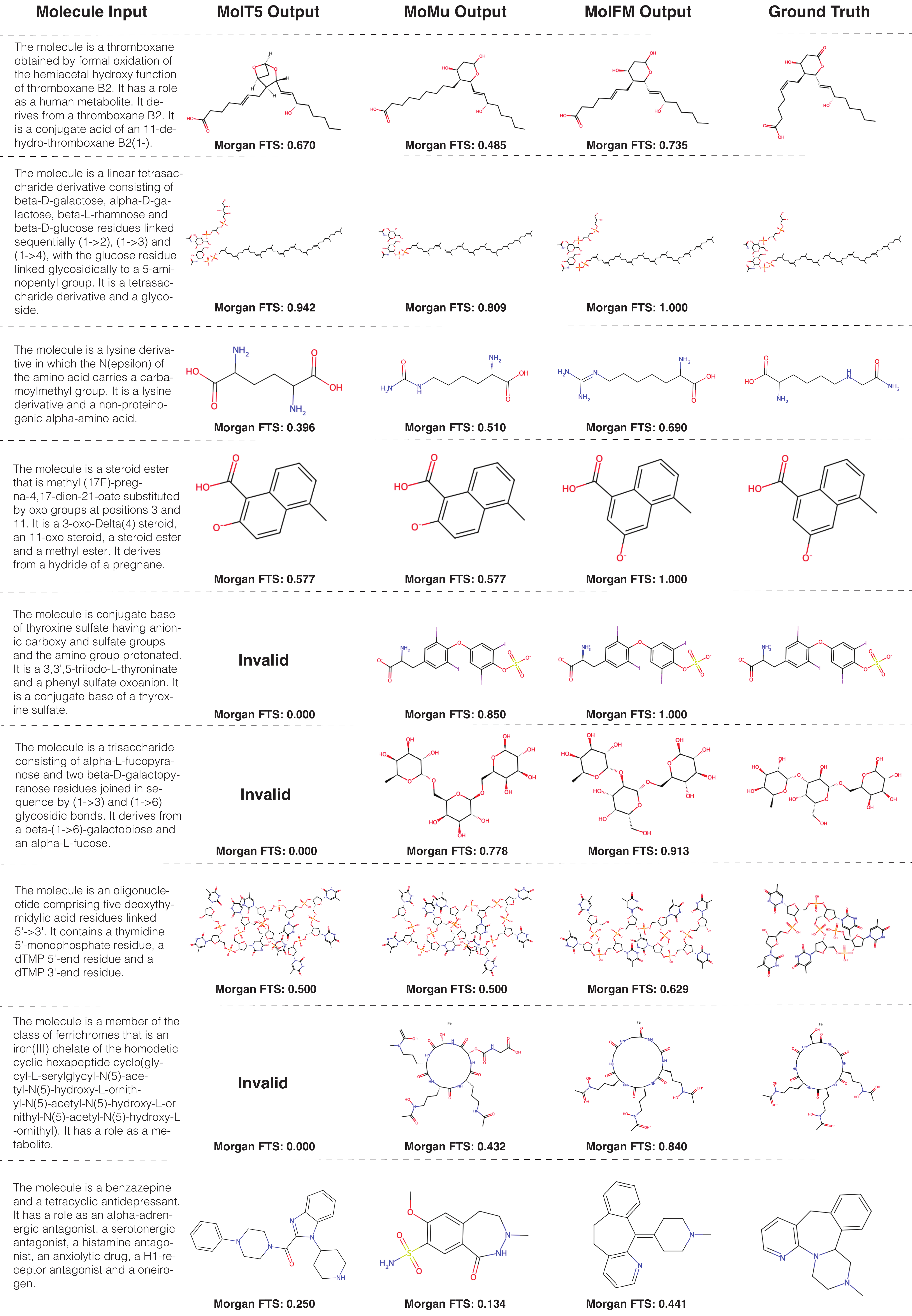}
\caption{Additional text-to-molecule generation cases. "Invalid" indicates that the generated SMILES can not be converted to a 2D molecular graph.}
\label{fig:text2smi_suppl}
\end{figure*}

\subsection{Additional visualization of cross-modal attention \label{sec:app_vis_attn}}
Fig. \ref{fig:attn_atom_appendix} shows the normalized cross-modal attention from texts to atoms. Fig. \ref{fig:attn_neigh_appendix} shows the normalized cross-modal attention from texts to neighbors in the knowledge graph.

\begin{figure*} [htpb!]
\centering
    \includegraphics[width=\linewidth]{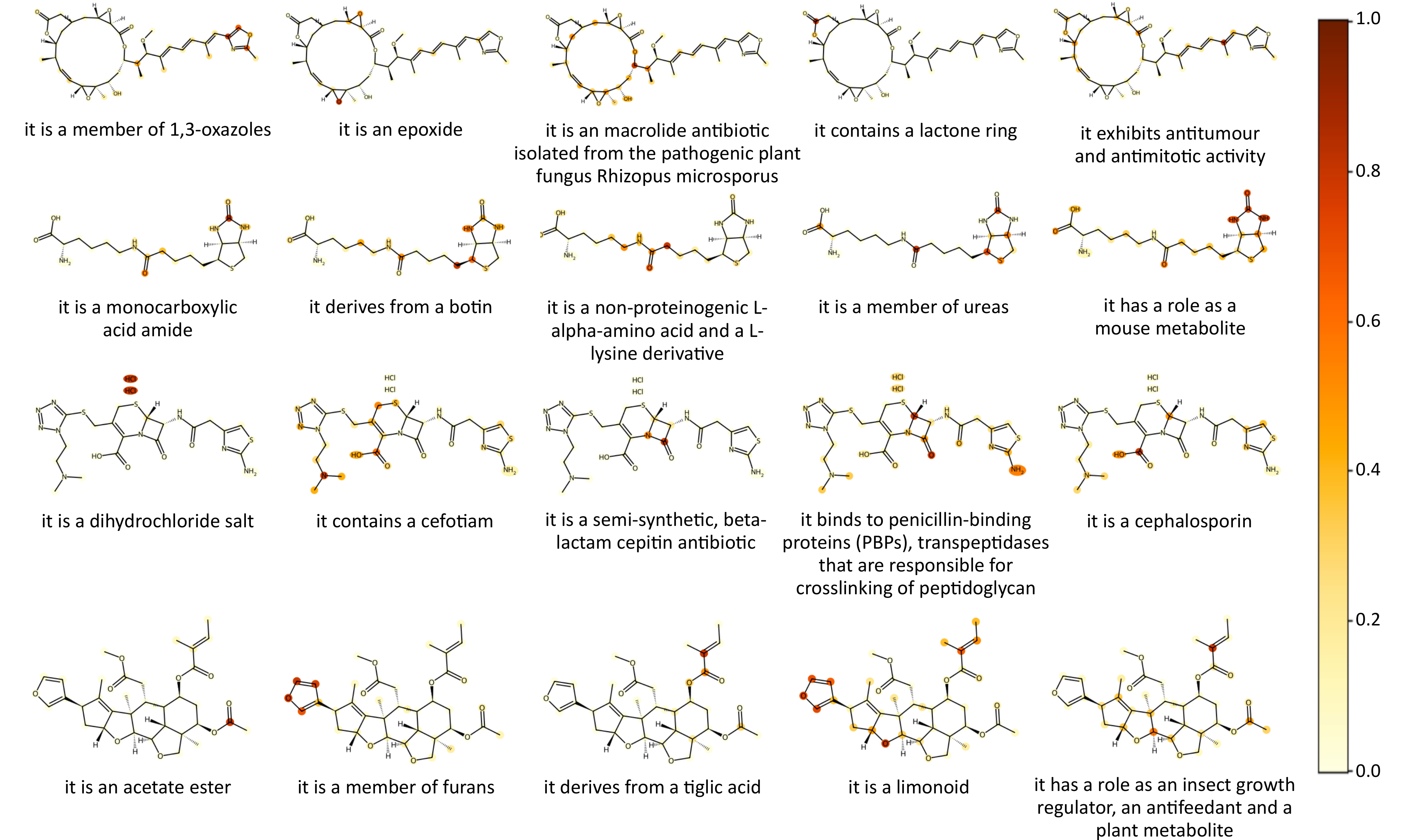}
\caption{Additional visualization of cross-modal attention from texts to atoms. }
\label{fig:attn_atom_appendix}
\end{figure*}

\begin{figure*} [htpb!]
\centering
    \includegraphics[width=\linewidth]{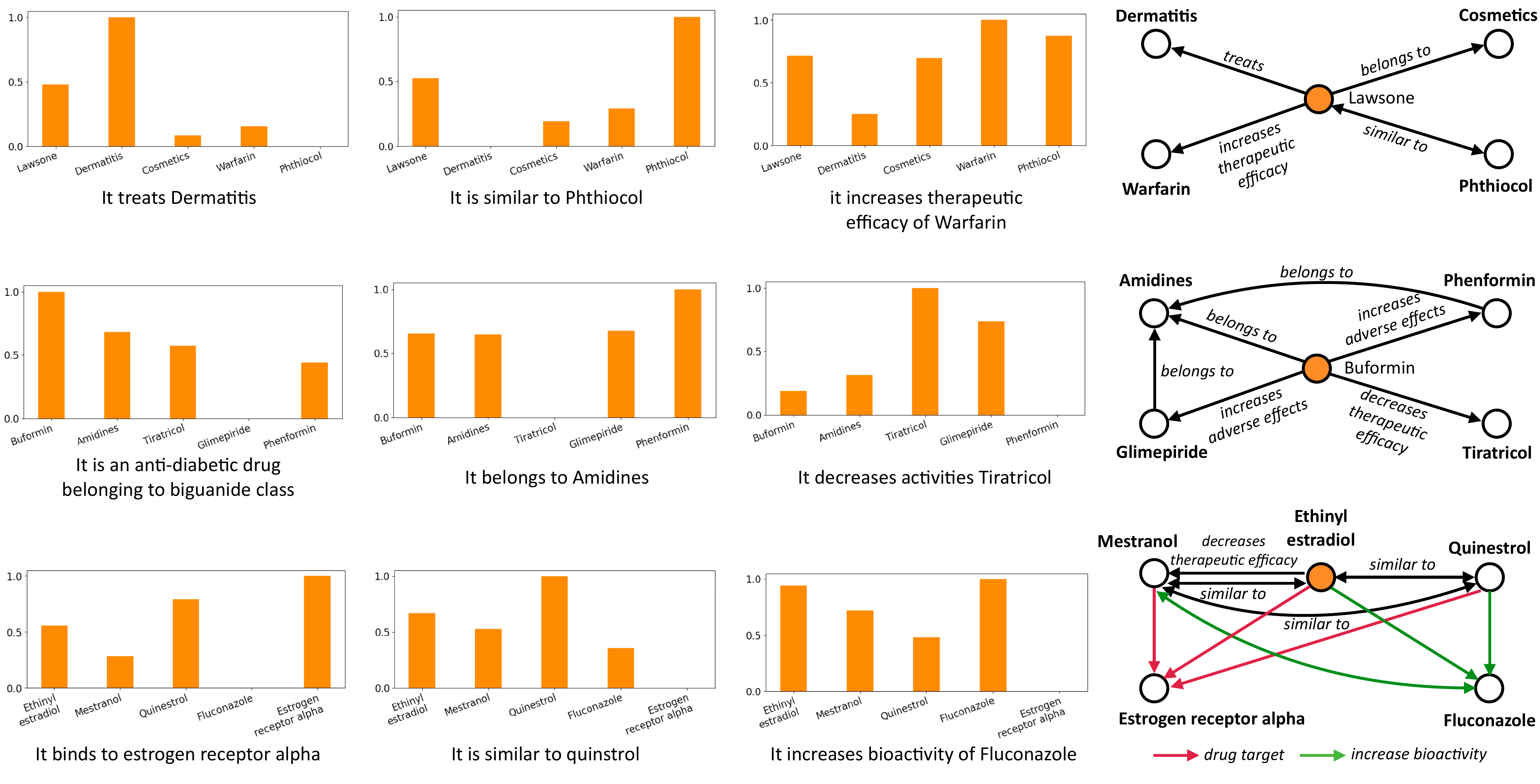}
\caption{Additional visualization of cross-modal attention from texts to neighbors. \textbf{Left}: the input text and the normalized attention to each entity. \textbf{Right}: the selected molecule (orange) and 4 randomly sampled neighboring entities, as well as relationships between these entities.}
\label{fig:attn_neigh_appendix}
\end{figure*}
\newpage

\bibliographyApp{refappendix}


\end{document}